\documentclass[10pt,journal,compsoc]{IEEEtran}

\ifCLASSOPTIONcompsoc
  \usepackage[nocompress]{cite}
\else
  \usepackage{cite}
\fi

\ifCLASSINFOpdf
\usepackage[pdftex]{graphicx}
\else
\fi

\usepackage{amsmath}
\usepackage{amsfonts}
\usepackage{amsthm}
\usepackage{mathtools}

\newcommand{\RR}{\mathbb{R}}
\newcommand{\NN}{\mathbb{N}}
\newcommand{\norm}[1] {\left \| #1 \right \|}
\newcommand{\Exp}[1] {\mathbb{E} \left[ #1 \right]}
\newcommand{\Prob}[1]{\text{\text{Pr}} \left[ #1 \right]}
\newcommand{\var}[1]{\text{Var} \left( #1 \right)}
\newcommand{\abs}[1]{\left| #1 \right|}
\newcommand{\Etl}[1] {\mathbb{E}_{{x}'(t,L)} \left[ #1 \right]}
\newcommand{\atrace}{\overline{\mathrm{Tr}} \text{ }} 
\newcommand{\trace}{\mathrm{Tr}  \text{ }}

\newtheorem{theorem}{Theorem}[section]
\newtheorem{lemma}[theorem]{Lemma}
\newtheorem{proposition}[theorem]{Proposition}
\newtheorem{corollary}[theorem]{Corollary}

\begin{document}

\title{On the Asymptotic Convergence of \\ Subgraph Generated Models}

\author{Xinchen~Xu,~Francesca~Parise
\thanks{X. Xu is with the Center for Applied Math, Cornell University, Ithaca, NY, 14850. F. Parise is with the School of Electrical and Computer Engineering, Cornell University, Ithaca, NY, 14850. E-mails: xx294@cornell.edu,  fp264@cornell.edu. This material is based upon work supported by the Air Force Office of Scientific Research under award number FA9550-24-1-0082. 
 This work has been submitted to the IEEE for possible publication. Copyright may be transferred without notice, after which this version may no longer be accessible. }

}

\IEEEtitleabstractindextext{%
\begin{abstract}
We study a family of random graph models - termed subgraph generated models (SUGMs) - initially developed by Chandrasekhar and Jackson in \cite{chandrasekhar2016network} in which higher-order structures are explicitly included in the network formation process. We use matrix concentration inequalities to show convergence of the adjacency matrix of networks realized from such SUGMs to the expected adjacency matrix as a function of the network size. We apply this result to study concentration of centrality measures (such as degree, eigenvector, and Katz centrality) in sampled networks to the corresponding centralities in the expected network, thus proving that node importance can be predicted from knowledge of the random graph model without the need of exact network data. 
\end{abstract}

}

\maketitle

\IEEEdisplaynontitleabstractindextext
\IEEEpeerreviewmaketitle

\IEEEraisesectionheading{\section{Introduction}\label{sec:introduction}}

 Many social and economic applications involve large populations of agents interacting in heterogeneous ways over a network. Consider, for instance, the dynamics of opinion exchange across social networks or the influence of peer decisions on an individual's choice to adopt a new product or behavior. The increasing size of this type of systems, exemplified by platforms like Facebook with billions of users, presents unique challenges for planners aiming to regulate these interactions. In fact in many cases, the planner cannot collect data about exact agents' interactions as this would be either too costly or impossible due to privacy or proprietary concerns \cite{breza2020using}. However, in these cases, it might be feasible for the planner to collect statistical information about agents' interactions that can be used to infer a random graph model. A key question is then whether knowledge of such a random graph model is sufficient to infer relevant  features of the realized network (or of a socio-economic process evolving over it). This question has been addressed in a number of recent works in the literature by focusing on random graph models in which each link is formed independently (such as Erdos-Renyi, stochastic block or graphon models).

As argued in \cite{chandrasekhar2016network}, many relevant networks nevertheless exhibit higher-order structure that cannot be captured by models in which links are realized independently from each other. For example, \cite{milo2002network} shows that the Bi-Fan network structure (Figure \ref{fig:intro}.A) is over-expressed in the C. Elegans neuronal network, and \cite{benson2016higher} uses the directed triangle structure (Figure \ref{fig:intro}.B) to discover social communities in the Twitter follower network. 

\begin{figure}[!h]
\centering
\includegraphics[width=2.5in]{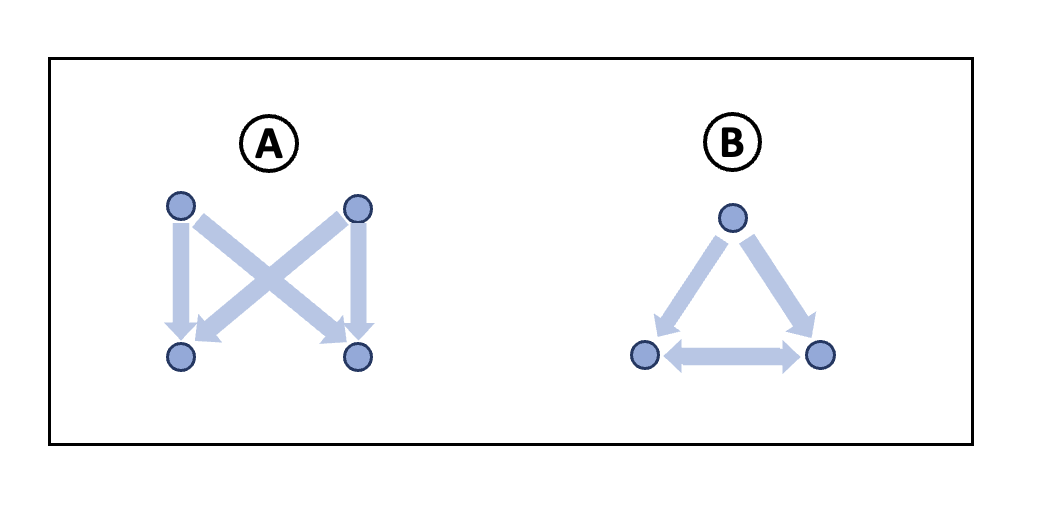}
\caption{Examples of higher-order structures found in real world networks.}
\label{fig:intro}
\end{figure}

To capture this type of complex interactions, \cite{chandrasekhar2016network} proposes a novel random graph model - the \textit{subgraph generated model (SUGM)} - in which one samples not only links but also higher order structures (such as, triangle or cliques). We here consider two variants of the SUGM (see Figure \ref{fig:formation}). First, we consider  the \textit{weighted} SUGM, in which the union of all generated subgraphs is used to construct a  network. The weight of each link in this network  corresponds to the number of times the link has been generated as part of different subgraphs realizations. Second, we consider the \textit{unweighted} SUGM (as originally introduced in \cite{chandrasekhar2016network}), in which links are either present or not, with no associated weight, and a link is present if it has been generated by at least one subgraph. For both models, our objective is to study whether key  properties of the realized networks can be inferred from knowledge of the generating process alone (the SUGM in our case). 

\subsection{Contributions}

As a first theoretical contribution, we derive an upper bound on the spectral norm of the difference between the realized and expected adjacency matrix, for both the weighted and unweighted case. The spectral norm (which coincides with the maximum eigenvalue for symmetric matrices) is relevant for studying a number of processes over networks. For the weighted SUGM, our concentration results follow straightforwardly from known matrix concentration inequalities since the adjacency matrix of the realized network can be rewritten as a sum of independent random matrices (each corresponding to one of the possible subgraphs). The analysis is instead more complex for the unweighted case, as keeping links that are generated in at least one of the subgraphs is a nonlinear operation. To overcome this issue we exploit matrix Efron-Stein inequalities to relate the unweighted adjacency matrices to a variance proxy. The main technical step is then to derive an upper bound on the log trace moment generating function of such variance proxy. 

As a second theoretical contribution, we bound the network centrality measures between the realized and expected networks. Specifically, we focus on degree, eigenvector and Katz centrality. Our interest for such measures stems from applications. For example, eigenvector centrality is related to the importance of each agent's initial opinion on the final consensus value in DeGroot opinion dynamics models \cite{golub2010naive}, while Katz centrality coincides with the Nash equilibrium in linear quadratic network games \cite{jackson2015games}. It is therefore relevant to study whether knowledge of the SUGM is sufficient to well approximate  these network measures, in the limit of large graphs. As the second main result, we provide an affirmative answer to this question under suitable assumptions on the probability with which each of the subgraphs is generated. Intuitively, we require the subgraphs to be generated frequently enough to have an increasing presence as the network grows, yet sparsely enough to guarantee that the chance of any given link being part of multiple subgraphs  vanishes. Under these assumptions we prove convergence (in normalized L1 norm) of the vector of centrality measures. 

Our work is mainly related to \cite{chandrasekhar2016network}. Therein the authors introduce the unweighted SUGM and focus on  statistical estimation of the generating parameters of the random graph model from network observations, while our paper focuses on the concentration of the realized networks to the expected networks. Notably, \cite{chandrasekhar2016network} relies on similar assumptions on the probabilities of the subgraph types as the ones needed in our work to ensure convergence of centrality measures. We also note that the SUGM is related to other random graph models based on subgraph distributions \cite{karrer2010random,ritchie2017generation}. In \cite{karrer2010random}, for example, a random graph model is proposed in which the realized graph is obtained as union of  subgraphs (similar to the SUGM weighted case), but the subgraph sampling procedure is different  as it  is based on a specified subgraph distribution (similar to the way links are sampled in the configuration model). Our paper is also related to a recent strand of literature in which random graph models (e.g., configuration, graphon or stochastic block models) are used to assess different network properties such as centrality measures (e.g. in \cite{avella2018centrality,dasaratha2020distributions}), opinion dynamics (e.g. in \cite{golub2012homophily}), equilibria of network games (e.g. in \cite{parise2023graphon,galeotti2010network,sadler2020diffusion}) and contagion processes (e.g. in \cite{rossi2017threshold,valente2020diffusion,akbarpour2020just,jackson2017behavioral}). None of the works cited above, however, focused on random graph models that capture higher-order structure, which instead is a main feature of the SUGM.

\subsection{Article Structure}

The rest of the paper is organized as follows. In Section 2 we provide the definitions of the two random network models this paper is built upon: the \textit{weighted} and the \textit{unweighted subgraph generated models}. Subsequently, in Section 3 we derive our main convergence results for the spectral norm in the two models. In Section 4, we apply the theoretical results to derive  convergence of graph centrality measures. Section 5 runs numerical simulations and demonstrates the theoretical results under several random network settings. Section 6 presents
concluding remarks and future directions. Appendix~A introduces common notations used in the proofs and Appendix~B contains the proofs omitted in the paper.

\section{Introduction of the Two Subgraph Generated Models}

We start with the definitions of the two random network models this paper is built upon: the \textit{weighted} and the \textit{unweighted subgraph generated model}. The latter was first introduced in \cite{chandrasekhar2016network} under the name of subgraph generated model (SUGM). We here introduce the prefix weighted/unweighted to distinguish two versions of this model. 

\begin{itemize}
        \item
        \textbf{Weighted Subgraph Generated Model (wSUGM)} \\ \\
        A Weighted Subgraph Generated Model of size $n$ is formally defined as follows. Consider a set $T$ of finitely many types of nonempty undirected subgraphs on which the model is based on; for instance, in a model with only links and triangles, we would have $T = \{\textit{link},\textit{triangle}\}$. The subgraphs are denoted by $(g_t(\cdot))_{t \in T}$, where each $g_t(L)$ is the subgraph of type $t$ generated on an  ordered  list\footnote{We hereby define $L$ as an ordered list of nodes so that subgraphs that are not complete or symmetric may be included in our generating process.} $L = [v_1, v_2, \ldots, v_{m_t}]$ of $m_t \leq n$ nodes, with $m_t$ denoting the size of the subgraph of type $t$. In the model $\text{wSUGM}(n,T,p)$, each possible subgraph $g_t(L)$ is formed independently with probability $p(t,L) \in [0,1]$, and the resulting adjacency matrix $A_w \in \NN^{n \times n}$ is formed by setting $(A_w)_{ij}$ equal to the \textbf{total number} of times edge $(i,j)$ is generated across all subgraphs involving edge $(i,j)$. \\

        \item
        \textbf{Unweighted Subgraph Generated Model (uSUGM)} \\ \\
        An Unweighted Subgraph Generated Model  of size $n$ is formally defined as follows. Let $T$, $(g_t(\cdot))_{t \in T}$ and $p(t,L)$ be defined as above. In the model $\text{uSUGM}(n,T,p)$, each possible subgraph $g_t(L)$ is formed independently with probability $p(t,L) \in [0,1]$, and the adjacency matrix $A_u \in \{0,1\}^{n \times n}$ is formed by setting $(A_u)_{ij} = 1$ if the edge $(i,j)$ is generated by \textbf{at least one} of the subgraphs involving edge $(i,j)$.  \\
\end{itemize}

\begin{figure*}[h]
\centering
\includegraphics[width=5.5in]{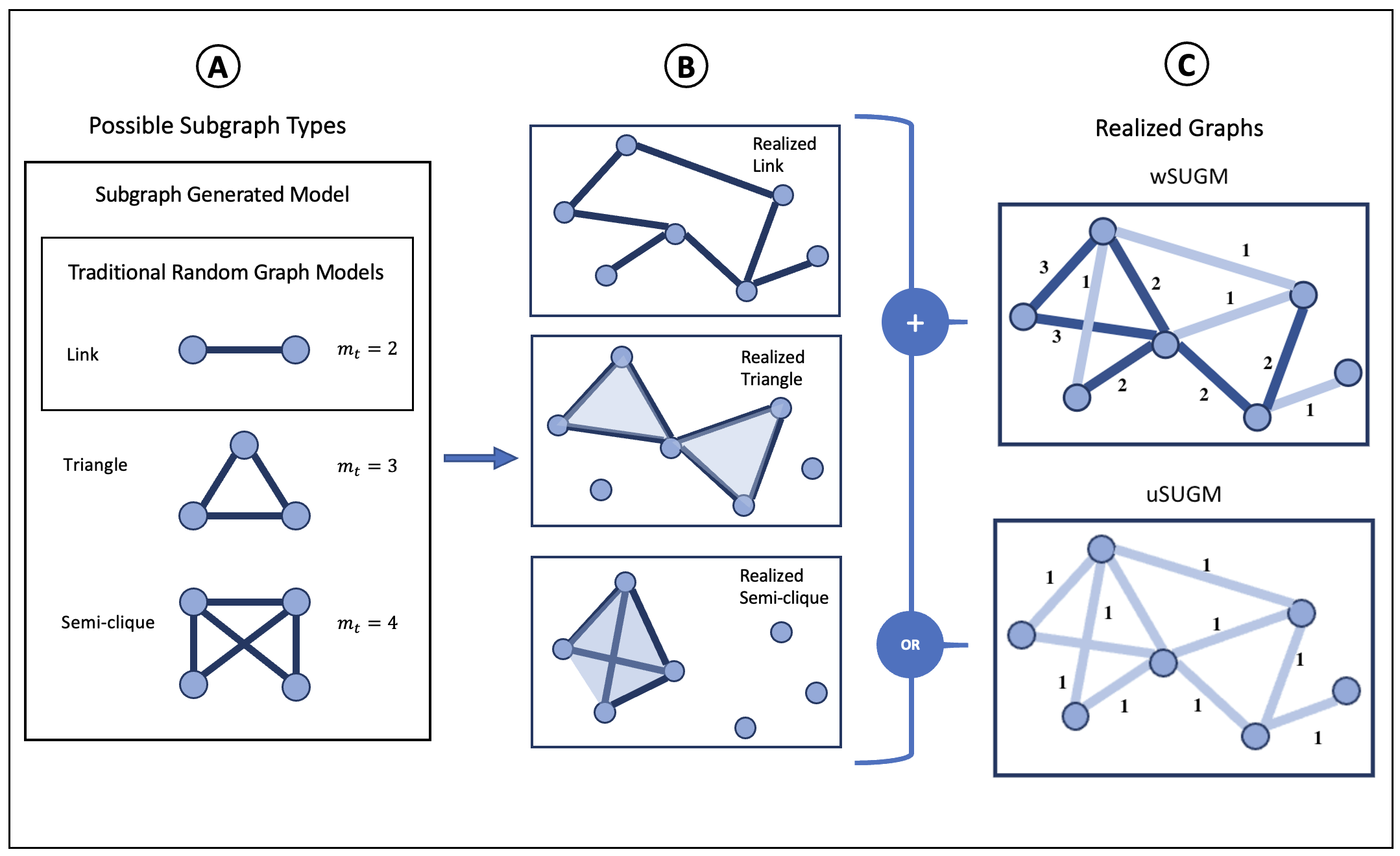} \\
\centering Fig. 2. Graph formation process for wSUGM and uSUGM.
\caption{In the SUGM, various subgraphs (e.g. links, triangles, semi-cliques) are generated independently according to given probability parameters. For the weighted SUGM, the resulting realized network is the direct sum of all such subgraph generated networks, while for the unweighted SUGM, the weight of each edge in the realized network can be viewed as applying the logic "OR" function over the weights of this edge in all realized subgraphs.}
\label{fig:formation}
\end{figure*}

Figure \ref{fig:formation}  describes the two graph generating process. These two random models are capable of generating networks that capture  the higher-order connectivity patterns observed in many real graphs. \cite{chandrasekhar2016network} showed that the unweighted subgraph generated model outperforms a number of random graph models in generating realistic distributions of networks, with fewer parameters.  

\section{Convergence Theory}

In this section, we present the main theorems of this paper. Specifically, we show that the spectral norm of the difference between the realized and expected  adjacency matrix of a SUGM can be bounded with high probability. 

\subsection{The Weighted Subgraph Generated Model} \label{section:wSUGM}

We start with presenting the result on the weighted model, which follows straightforwardly from standard matrix concentration inequalities. 

\begin{proposition} \label{prop:main_wSUGM_easy}
    Let $G_w$ be a random graph of size $n$ generated by the $\text{wSUGM}(n,T,p)$ with finite subgraph type set $T$ and probabilities $p(\cdot,\cdot)$. Let $A_w$ be the adjacency matrix of $G_w$. Denote the maximum expected degree by $\Delta_w := \norm{\Exp{A_w}}_{\infty}$  and the max subgraph size by $M := \max_{t \in T} \{m_t\}$. Let $\epsilon > 0$, and suppose that for $n$ sufficiently large,
    \begin{equation*} \label{Assumption:P1_A1}
        \Delta_w> \frac{4}{9} \ln{(2n/\epsilon)} \tag{\text{A1}}.
    \end{equation*}
    Then with probability at least $1-\epsilon$, for $n$ sufficiently large, 
    $$\norm{A_w - \Exp{A_w}}_2 \leq \sqrt{4M^2 \Delta_w\ln{(2n/\epsilon)}}.$$
\end{proposition}

Here is a brief summary of the techniques used in the proof of this proposition. For a fixed size $n$, consider $\text{wSUGM}(n,T,p)$: for each subgraph type $t$ and ordered list $L$ of $m_t$ vertices, we construct the Bernoulli random variable $x(t,L) \sim \textit{Bern}(p(t,L))$. 
If we enforce an ordering on the random variables $x(t,L)$ by assigning an order to the subgraph types $t \in T$ and  then sorting $L$ lexicographically, then the entries of the random vector 
\begin{align*}
    \Vec{x} = \{ \ldots, x(t,L), \ldots \} \in \mathcal{Z}
\end{align*}
are mutually independent random variables $x(t,L) \in [0,1]$. Construct $A{(t,L)} \in \{0,1\}^{n \times n}$ as the adjacency matrix representing the subgraph $g_t(L)$: $A{(t,L)}_{ij} = A{(t,L)}_{ji} = 1$ if and only if $i \in L, j \in L$ and edge $(i,j)$ exists in $g_t(L)$. Define the measurable function $W : \mathcal{Z}\rightarrow \mathrm{H}^n$ where $\mathrm{H}^n$ is the set of $n \times n$ symmetric matrices and 
$$W(\Vec{x}) := \sum_{t \in T} \sum_{L \in \mathfrak{S}(n, m_t)}  x(t,L) A(t,L)$$
where the notation $\mathfrak{S}(n, m_t)$ contains all possible  ordered lists of size $m_t$ over $n$ vertices. Note that $W(\Vec{x})$ constructs the adjacency  matrix of the realized network. Since $W(\Vec{x})$ is composed of a summation of independent symmetric matrices, $X(t,L) := x(t,L) A(t,L)$, we can apply the following matrix concentration inequality result. 

\begin{theorem} [Theorem 5, \cite{chung2011spectra}] \label{thm:chung_ineq}
 Consider $m$ zero-mean independent random symmetric matrices $X_1, \ldots, X_m$ of dimension $n$. If $\norm{X_i}_2 \leq K$ for all $i$, then for any $a > 0$,  
$$\Prob{\norm{\sum_{i = 1}^m X_i}_2 > a} \leq 2 n \cdot \exp \left( -\frac{a^2}{2v^2+2Ka/3} \right)$$
where $v^2 := \norm{\sum_{i = 1}^m \var{X_i}}_2$.
\end{theorem}

To apply this theorem to $W(\Vec{x})$, we need to derive a bound on $\norm{X(t,L) - \Exp{X(t,L)}}_2$ for each random matrix $X(t,L)$ and an estimate of the variance $\norm{\sum_{t \in T} \sum_{L \in \mathfrak{S}(n, m_t)} \var{X(t,L)} }_2$. This is done in Appendix \ref{appendix:main_wSUGM_easy}. 

\subsection{Setup for Efron-Stein Inequalities} \label{section:ES_setup}

In the $\text{uSUGM}(n,T,p)$, the adjacency matrix $A_u \in \{0,1\}^{n \times n}$ is formed by setting $(A_u)_{ij}=(A_u)_{ji} = 1$ if the edge $(i,j)$ is generated by \textbf{at least one} of the subgraphs involving edge $(i,j)$. To model this, define the function $\omega : \mathrm{H}^n \rightarrow \mathrm{H}^n$ as: 
$$\omega(A)_{ij} = \min(A_{ij}, 1)$$
and the  measurable function $U : \mathcal{Z} \rightarrow \mathrm{H}^n$ as
    \begin{align*}
        U(\Vec{x}) &= \omega\left( W(\Vec{x}) \right) \\
        &= \omega\left(\sum_{t \in T} \sum_{L \in \mathfrak{S}(n, m_t)} x(t,L) A(t,L)\right), 
    \end{align*}
so that $A_u= U(\Vec{x}) $. 

To handle  the non-linearity involved in $U$, we will be using the matrix Efron–Stein inequalities developed by \cite{paulin2016efron}.
To this end, note that for a fixed $n$,  $\Exp{\norm{U(\Vec{x})}_2} < \infty$ and let 
$$\hat{U} := \hat{U}(\Vec{x}) := U(\Vec{x}) - \Exp{U(\Vec{x})}$$ be the corresponding centered random matrix. 

Recall the random vector 
\begin{align*}
    \Vec{x} = \{ \ldots, x(t,L), \ldots \}
\end{align*}
contains mutually independent Bernoulli random variables $x(t,L) \in [0,1]$. For each coordinate $(t,L)$, we can construct another random vector 
    $$\Vec{x}^{(t,L)} = ( \ldots, {x}'(t,L), \ldots )$$
where ${x}'(t,L)$ is an independent copy of ${x}(t,L)$ and nothing else is changed. Note that $\Vec{x}$ and $\Vec{x}^{(t,L)}$ follow the same distribution and  only differ by the coordinate $(t,L)$. Using $\Vec{x}^{(t,L)}$ we form the random matrices
\begin{align*}
    \hat{U}^{(t,L)} := \hat{U}(\Vec{x}^{(t,L)}) = U(\Vec{x}^{(t,L)}) - \Exp{U(\Vec{x}^{(t,L)})}
\end{align*}
for all $t \in T$, $L \in \mathfrak{S}(n, m_t)$.  Note that since the random matrix $U(\Vec{x}^{(t,L)})$ is identically distributed to $U(\Vec{x})$,
$\Exp{{U}(\Vec{x}^{(t,L)})} = \Exp{{U}(\Vec{x})}$. 
The matrix Efron–Stein inequality \cite{paulin2016efron} bounds the trace moments of the centered random matrix $\hat{U}(\Vec{x})$ in terms of the moments of its variance proxy
\begin{align*}
    V_U(\Vec{x}) := \frac{1}{2} \sum_{t \in T} \sum_{L \in \mathfrak{S}(n, m_t)} \Etl{(\hat{U}-\hat{U}^{(t,L)})^2 | \Vec{x}} .
\end{align*}
 as summarized next.

\begin{theorem} [Theorem 4.3, \cite{paulin2016efron}] \label{thm:es}
    Assume that the centered random matrix $\hat{U}(\Vec{x})$ is bounded. Then, for any $\psi > 0$ and $|\theta| \leq \sqrt{\psi/2}$, 
    \begin{align*}
        \log \Exp{\atrace e^{\theta \hat{U}}} \leq \frac{\theta^2/\psi}{1-2\theta^2 / \psi}\log \Exp{ \atrace e^{\psi V_U}},
    \end{align*}
     where $\atrace (\cdot):= \frac{1}{n} \trace (\cdot)$ is the normalized trace.
\end{theorem}

Note that the left hand side of the inequality in Theorem \ref{thm:es} is the log of the trace moment generating function (m.g.f.) of a centered random matrix $\hat{U}$: 
    $$m_{\hat{U}}(\theta) := \Exp{\atrace \exp (\theta \hat{U})}.$$
By the matrix Laplace transform method, this can be used to bound the norm of $\hat U$, as desired.

\begin{theorem} [Proposition 3.3, \cite{mackey2014matrix}] \label{thm:matrix_lap}
    Let $X \in \mathrm{H}^n$ be a centered random matrix with trace m.g.f. $m_{X}(\theta) = \Exp{\atrace \exp (\theta X)}$. For each $t \in \RR$, 
    \begin{align*}
        \Prob{\lambda_{\max}(X) \geq t} &\leq n \cdot \inf_{\theta > 0} \exp \{ - \theta t + \log m_{X}(\theta) \}, \\
        \Prob{\lambda_{\min}(X) \leq t} &\leq n \cdot \inf_{\theta < 0} \exp \{ - \theta t + \log m_{X}(\theta) \}.
    \end{align*}
\end{theorem}

To summarize, if we can find an upper bound on $\log \Exp{ \atrace e^{\psi V_U}}$ on the right hand side of the inequality in Theorem \ref{thm:es}, then combining it with the two bounds in Theorem \ref{thm:matrix_lap}, we  obtain a  probability bound on $\norm{\hat{U}(\Vec{x})}_2$, which is  the spectral norm of the difference between the realized and expected adjacency matrix in the unweighted SUGMs.\footnote{Recall that for a symmetrix matrix $X$, $\|X\|_2\le t$ if and only if $\lambda_{\max}(X) \leq t$ and $\lambda_{\min}(X) \geq -t$. } This is the result discussed in the next section.

\subsection{The Unweighted Subgraph Generated Model (uSUGM)} \label{section:uSUGM}

\begin{proposition}\label{prop:main_uSUGM_ES}
    Let $G_u$ be a random graph of size $n$ generated by the $\text{uSUGM}(n,T,p)$ with subgraph type set $T$ and probabilities $p(\cdot,\cdot)$. Let $A_u$ be the adjacency matrix of $G_u$. Denote $\Delta_u := \norm{\Exp{A_u}}_{\infty}$ as the maximum expected degree, and $M := \max_{t \in T} \{m_t\}$ as the max subgraph size. Let $\epsilon > 0$, suppose that for $n$ sufficiently large,
    \begin{equation*} \label{Assumption:P2_A2}
        \frac{\Delta_w}{\Delta_u} \leq \mu_{T,p} \tag{A2}
    \end{equation*}
    for some scalar $\mu_{T,p}$  independent of the size $n$, and 
    \begin{equation*} \label{Assumption:P2_A3}
        \Delta_u > \frac{16}{\mu_{T,p}} \ln{(2n/\epsilon)} \tag{A3}.
    \end{equation*}
    Then with probability at least $1-\epsilon$, for $n$ sufficiently large, 
    $$\norm{A_u - \Exp{A_u}}_2 \leq 4M^2\sqrt{\mu_{T,p} \Delta_u \ln{(2n/\epsilon)}}.$$
\end{proposition}

We next present a summary of how the proof of this proposition is constructed. The full proof is detailed in Appendix \ref{appendix:main_uSUGM_ES}. For a fixed size $n$, finite set of subgraph types $T$, and probabilities $p(\cdot,\cdot)$, we will be using both the $\text{uSUGM}(n,T,p)$ and the $\text{wSUGM}(n,T,p)$. 

First, we construct the variance proxy for the weighed SUGM following a similar procedure as in Section \ref{section:ES_setup}. To this end, recall that 
$$W(\Vec{x}) = \sum_{t \in T} \sum_{L \in \mathfrak{S}(n, m_t)} x(t,L) A(t,L)$$
 constructs the adjacency matrix of the realized network from the wSUGM. From this, we can  construct the centered random matrices 
\begin{equation}\label{W}\begin{aligned}
    \hat{W} &:= \hat{W}(\Vec{x}) := W(\Vec{x}) - \Exp{W(\Vec{x})} \\
    \hat{W}^{(t,L)} &:= \hat{W}(\Vec{x}^{(t,L)}) = W(\Vec{x}^{(t,L)}) - \Exp{W(\Vec{x}^{(t,L)})}
\end{aligned}\end{equation}
and the  variance proxy for the weighted case
\begin{align*}
    V_W(\Vec{x}) := \frac{1}{2} \sum_{t \in T} \sum_{L \in \mathfrak{S}(n, m_t)} \Etl{(\hat{W}-\hat{W}^{(t,L)})^2 | \Vec{x}},
\end{align*}
as done in Section \ref{section:ES_setup} for the unweighted case.

The variance proxies for the weighted and unweighted cases satisfy the following relation.

\begin{lemma} \label{lem:VRVS_entry}
    For any $\Vec{x} \in \mathcal{Z}$, we have
    \begin{align*}
        0 \leq_{\text{e}} V_U(\Vec{x}) \leq_{\text{e}} V_W(\Vec{x})
    \end{align*}
    where $\leq_{\text{e}}$ stands for the entry-wise $\leq$ relationship.
\end{lemma}

This allows us to upper bound the log trace m.g.f. of the unweighted variance proxy, $\log \Exp{ \atrace e^{\psi V_U}}$, with the log trace m.g.f. of the weighted variance proxy, $\log \Exp{ \atrace e^{\psi V_W}}$, which is easier to study. 

\begin{lemma} \label{lem:VRVS_overall}
    \begin{align*}
        \log \Exp{ \atrace e^{\psi V_U}} \leq \log \Exp{ \atrace e^{\psi V_W}}.
    \end{align*}
\end{lemma}

We next derive an upper bound on the weighted variance proxy using properties of the matrix logarithm and matrix exponential operators. 

\begin{lemma} \label{lem:bound_weight_proxy}
For any $0<\psi < \frac{1}{M^2}$, we have 
\begin{align*}
    \log \Exp{ \atrace e^{\psi V_W}} &\leq (e-1) \Delta_w M  \cdot \psi
\end{align*}
where $\Delta_w := \norm{\Exp{A_w}}_{\infty}$.
\end{lemma}

Lastly, we combine all the lemmas together and apply Theorem \ref{thm:matrix_lap}. This is done in Appendix~\ref{appendix:main_uSUGM_ES}, thus concluding the proof of Proposition \ref{prop:main_uSUGM_ES}. 

It is worth noticing that as a byproduct of this proof, we find another way to show  convergence for the wSUGM.

\begin{corollary}\label{prop:main_wSUGM_ES}
    Let $G_w$ be a random graph of size $n$ generated by the $\text{wSUGM}(n,T,p)$ with subgraph type set $T$ and probabilities $p(\cdot,\cdot)$. Let $A_w$ be the adjacency matrix of $G_w$. Denote $\Delta_w := \norm{\Exp{A_w}}_{\infty}$ as the maximum expected degree, and $M := \max_{t \in T} \{m_t\}$ as the max subgraph size. Let $\epsilon > 0$, suppose that for $n$ sufficiently large,
    \begin{equation*} \label{Assumption:P3_A4}
        \Delta_w > 16 \ln{(2n/\epsilon)} \tag{A4}.
    \end{equation*}
    Then with probability at least $1-\epsilon$, for $n$ sufficiently large, 
    $$\norm{A_w - \Exp{A_w}}_2 \leq 4M^2\sqrt{\Delta_w \ln{(2n/\epsilon)}}.$$
\end{corollary}

By comparing Proposition \ref{prop:main_wSUGM_easy} with Corollary \ref{prop:main_wSUGM_ES}, one can see that although the two bounds have the same asymptotic order, the Efron-Stein inequality method results in a worse leading coefficient due to the use of less tight relaxations while constructing the bounds. 

\subsection{Sufficient Conditions on Subgraph Generating Probabilities}\label{section:prob}

In this section, we derive sufficient conditions for Assumptions \eqref{Assumption:P1_A1}, \eqref{Assumption:P2_A2}, \eqref{Assumption:P2_A3} and \eqref{Assumption:P3_A4} to hold in terms of the subgraph generating probabilities. We start by assuming that, for each subgraph type $t \in T$, each subgraph $g_t(L)$ with a non-zero generating probability has a generating probability of the same order of magnitude with respect to $n$, that is 
\begin{multline}\label{Assumption:prob_A5}
    p(t,L) = 0 ~\text{, or }~ p(t,L) = \frac{b_t(L)}{n^{h_t}}, \\
    ~\textup{ with }~   m_t -2 < h_t < m_t-1 \text{, } b_t(L) \in [l_t,u_t] \tag{\text{A5}}
\end{multline} 
for every $L \in \mathfrak{S}(n,m_t)$ where $0 < l_t < u_t $ are some fixed constants.

We also assume that for each subgraph type $t \in T$, (1) The proportion of subgraphs $g_t(L)$ that have non-zero generating probability $p(t,L) > 0$ is lower bounded; and (2) The proportion of subgraphs $g_t(L)$ containing nodes $i,j \in L$ that have non-zero generating probability $p(t,L) > 0$ is lower bounded, for each pair of nodes $i,j \in V$. That is,
\begin{align*}\label{Assumption:prob_A6}
    ~~&\exists ~\xi_t \in (0,1)  \text{ s.t } \\
    (1) ~~ & |\{L \in \mathfrak{S}(n,m_t): p(t,L) > 0\}| \geq \xi_t \cdot |\mathfrak{S}(n,m_t)| \\
    (2) ~~&\forall~ i,j \in V,  |\{L \in \mathfrak{S}(n,m_t): i,j \in L ~,~ p(t,L) > 0\}| \\
    &~\quad\quad\geq \xi_{t} \cdot |L \in \mathfrak{S}(n,m_t): i,j \in L|. \tag{A6}
\end{align*}

We next show that assumptions \eqref{Assumption:prob_A5} and \eqref{Assumption:prob_A6} are sufficient to guarantee all previous assumptions. We provide two brief insights into the above inequalities before presenting the main lemma of this section. Let $G$ be a random graph of size $n$ generated by the $\text{SUGM}(n,T,p)$ (either weighted or unweighted):
\begin{enumerate}
    \item On one hand, the expected degree of any node in the weighted model is upper bounded by 
    \begin{align*}
        &~~~~\sum_{t \in T} {n-1 \choose m_t-1} m_t! (m_t-1) \frac{u_t}{n^{h_t}} \\
        &\leq \sum_{t \in T} {n^{m_t-1}} m_t(m_t -1 ) \frac{u_t}{n^{h_t}}\\
        &< \sum_{t \in T} n^{m_t - h_t -1} u_t m_t^2.
    \end{align*}
    Therefore, the left condition in \eqref{Assumption:prob_A5} ensures that the overall degree of any node grows linearly or sub-linearly. 

    \item On the other hand, the expected total number of subgraphs in the entire network can be lower bounded by 
    \begin{align*}
         \sum_{t \in T} \sum_{L \in \mathfrak{S}} p(t,L) &\geq  \sum_{t \in T}  \xi_t {n \choose m_t}  \frac{l_t}{n^{h_t}} \\
         & \geq \sum_{t \in T} \xi_t \left(\frac{n}{m_t}\right)^{m_t} \frac{l_t}{n^{h_t}} \\
         &= \sum_{t \in T} \xi_t n^{m_t - h_t} \frac{l_t}{{m_t}^{m_t}},
    \end{align*}
    where we used $ {n \choose k}\ge \left(\frac{n}{k}\right)^k.$
    Therefore, the right condition in \eqref{Assumption:prob_A5} ensures that the expected number of subgraphs in the entire network  grows superlinearly.
\end{enumerate}

\begin{lemma}  \label{lem:prob_frac}
    If Assumption \eqref{Assumption:prob_A5} and \eqref{Assumption:prob_A6} hold, then 
    \begin{align*}
        \text{Assumption } &\eqref{Assumption:P1_A1}\ and \ \eqref{Assumption:P3_A4} \text{ hold for wSUGM$(n,T,p)$, and } \\
        \text{Assumption } &\eqref{Assumption:P2_A2}\ and \ \eqref{Assumption:P2_A3} \text{ hold for uSUGM$(n,T,p)$.}
    \end{align*}
\end{lemma}

\section{Application To Graph Centrality Measures}

\label{section:application}

A key task in complex system analysis is the identification of key nodes or agents in a network. To tackle this task, different graph centrality measures have been proposed to quantify node importance. In our analysis, we  focus on the following commonly used centrality measures.

\begin{itemize}
        \item \textbf{Degree centrality} $c^d (\cdot)$ measures the local importance of a node $i$ based on the number of neighbors. Mathematically,
        \begin{equation*} \label{defn:degree_centrality}
            c^d (A) := A \mathbf{1}.
        \end{equation*}
        
        \item \textbf{Eigenvector centrality} $c^e_n(\cdot)$  considers not only on the number of neighbors, but also their importance.
         Mathematically, eigenvector centrality is defined as 
        \begin{equation*} \label{defn:eigen_centrality}
            c^e := \sqrt{n}{v_1} 
        \end{equation*}
        where $v_1$ is the dominant eigenvector\footnote{If $A$ is symmetric and if the associated graph is connected, the eigenvalues of $A$ are real and can be ordered as $\lambda_1 > \lambda_2 \geq \ldots \geq \lambda_n$, with $\lambda_1$ being a simple eigenvalue by the Perron-Frobenious theorem.} of $A$ normalized to have unit norm $\norm{v_1}_2 = 1$ and the scaling factor $\sqrt{n}$ is needed to guarantee that the centrality will not tend to zero with increasing graph size, see \cite{avella2018centrality}. \\

        \item \textbf{Katz centrality} $c^k_{\alpha} (\cdot)$ measures the importance of a node based on neighbors that are multiple-hops away, discounted by a weighting factor $\alpha\in (0,1)$. Mathematically, it is computed as 
        $$c^k_{\alpha} (A) = \left( \sum_{i = 0}^{\infty} (\alpha A)^i \mathbf{1} \right). $$ 
         By choosing  $\alpha \in (0,\frac{1}{\lambda_1})$, the series converges and  
        \begin{equation*} \label{defn:katz_centrality}
            c^k_{\alpha} (A) := (I - \alpha A)^{-1} \mathbf{1}.
        \end{equation*}
\end{itemize}

Before delving in our analysis of these measures for networks sampled from a SUGM, we present two motivating examples illustrating their importance  for different socio-economic systems. 

\begin{enumerate}
    \item
    \textbf{DeGroot Opinion Dynamics in Influence Networks.}\cite{golub2010naive} \\
    Consider a network $G$ consisting of $n$ agents where each agent $i$ has an initial opinion $p_i(0) \in [0,1]$ on a subject. Agents update their opinions at each time by communicating with their neighbors and taking weighted averages of their neighbors’ opinions from the previous period. Let $A$ be the weighted adjacency matrix of the network $G$. The update rule can be formulated as the following equation, 
    $$p_i(t+1) = \frac{1}{\sum_{j} A_{ij}} \sum_{j} A_{ij} p_j(t).$$
    When the network $G$ is strongly connected and aperiodic,  the DeGroot opinion dynamics converge to a unique consensus 
    $$\lim_{t \rightarrow \infty} \vec{p}(t) = (\vec{w}^T \vec{p}(0)) \cdot \mathbf{1}$$
    where $\vec{w} \geq 0$, $\norm{\vec{w}}_1 = 1$ is the left dominant eigenvector of the row-stochastic normalized adjacency matrix. Note that the networks generated by the SUGM are undirected, hence the left dominant eigenvector coincides with  eigenvector centrality. \\
        
    \item
    \textbf{Linear Quadratic Network Games.} \cite{jackson2015games} \\
    Consider a network $G$ consisting of $n$ agents where each agent $i$ plays a simultaneous game by choosing action $a_i \in \RR_+$ to maximize a linear quadratic utility function
    $$u_i = - \frac{1}{2} a_i^2 + a_i \left( b_i + \frac{\beta}{\sum_{j} A_{ij}} \sum_{j} A_{ij} a_j \right)$$
    where $b_i \in \RR_+$ represents the standalone heterogeneous marginal return on each agent's individual action and $\beta \in \RR_+$ is a parameter capturing the strength of peer effects.\\
    As pointed out in \cite{ballester2006s}, when the spectral radius of the matrix $\frac{\beta}{\sum_{j} A_{ij}} A=: \beta \tilde A$ is strictly less than $1$, there exists a unique Nash equilibrium action 
    $$\vec{a}= \left(I - \beta \tilde A\right)^{-1} \vec{b}$$
    where $\vec{b} \in \RR^n$ contains the scalar $b_i$'s. 
    For $\vec{b}=\mathbf{1}$, that is when all agents have the same standalone marginal return, this equilibrium action coincides with Katz centrality (for $\alpha = \beta$). 
\end{enumerate}

Note that in both examples, some form of normalization is needed in order for the local aggregate information to remain bounded as the network size grows. Therefore in this section, we will analyze the  normalized graph $\Bar{A} := A/\Delta$, where $\Delta := \norm{\Exp{A}}_{\infty}$.
We will also make the following assumption 
\begin{align*} \label{Assumption:connect}
    &~~~~~\Bar{A} \text{ corresponds to a connected graph} \\
    &\text{where } \lambda_1(\Exp{\Bar{A}}) \rightarrow \lambda^* \text{ as } n \rightarrow \infty \\
    &\text{and } \abs{\lambda_1(\Exp{\Bar{A}}) - \lambda_2(\Exp{\Bar{A}})} \in [m,M]. \tag{\text{A7}}
\end{align*} 
where $\lambda_i(\Exp{\Bar{A}})$ corresponds to the $i-th$ largest eigenvalue of the expectation matrix $\Exp{\Bar{A}}$ and $m,M > 0$ are fixed constants. Notice that by the Perron-Frobenious theorem, $\Bar{A}$ is irreducible by \eqref{Assumption:connect} and $\lambda_1(\Bar{A})$ is a simple eigenvalue. This assumption guarantees that the centrality measures are well defined and is critical in the convergence analysis. 

Using Propositions \ref{prop:main_wSUGM_easy} and \ref{prop:main_uSUGM_ES}, we show that the average distance between the centrality measures of nodes in networks sampled from the SUGM  converges to the centrality measures of the corresponding nodes in the expected network with high probability.

\begin{corollary} [Convergence of graph centrality measures] \label{cor:centrality}
    Let $G$ be a random graph of size $n$ generated by the $\text{SUGM}(n,T,p)$ (either weighted or unweighted). Let $A$ be the adjacency matrix of $G$, $\Delta := \norm{\Exp{A}}_{\infty}$ be the maximum expected degree, and $\Bar{A} := A/\Delta$ be the normalized adjacency matrix. Let $c(\cdot)$ be either degree centrality, eigenvector centrality or Katz centrality (with $\alpha \in (0, 1/\lambda^*)$). Suppose that Assumptions~\eqref{Assumption:prob_A5}, \eqref{Assumption:prob_A6} and \eqref{Assumption:connect} are satisfied, then with probability at least $1-\epsilon$,  
    \begin{align*}
        \frac{1}{n} \norm{c(\Bar{A}) - c(\Exp{\Bar{A}})}_1 &\leq \phi(n)
    \end{align*}
    where $\phi(n) \rightarrow 0$ as $n \rightarrow \infty$.
\end{corollary}

This convergence result is of practical importance for settings in which collecting exact network data may be too costly. In fact, Theorem \ref{cor:centrality} guarantees that, in the limit of large networks, one can use information about the generating process (which is typically easier  to obtain \cite{breza2020using}) to predict node importance in any realized network without the need for exact data. 

\section{Numerical Simulations}

To validate our theoretical contributions, we implemented a suite of numerical simulations for SUGM involving links and triangles. These simulations encompassed three distinct probability models from which the subgraphs can be generated: the uniform model, the stochastic block model (SBM), and a distance-based model.

\begin{itemize}
    \item
    \textbf{The uniform model.} \\ 
    This model generates random graphs where the probability of forming links and triangles is uniform across all pairs and triplets of nodes. Specifically, the link probability is defined as $p({\text{link}}) = \frac{5}{n^{0.65}}$ and the triangle probability as $p({\text{triangle}}) = \frac{1}{n^{1.4}}$, where $n$ is the size of the graph. \\

    \item
    \textbf{The stochastic block model model (SBM).} \\ 
    This model assumes that nodes belong to either one of two communities (with $70-30\%$ split). Link probability between nodes of the same community is higher ($\frac{7}{n^{0.65}}$) compared to the link probability between nodes of different communities ($\frac{2}{n^{0.65}}$). The same applies to triangles. Specifically, the probability of generating a triangle in which each of the three nodes is within the same community is higher ($\frac{1}{n^{1.4}}$) than triangles with nodes from mixed communities ($\frac{0.1}{n^{1.4}}$). \\
        
    \item
    \textbf{The distance-based model.} \\
    In this model, each node is equidistantly assigned a position between $0$ and $1$ and the probability of link and triangle formation is influenced by the relative positions or 'distances' between nodes. The link probability is   $p({\text{link}}) = - C_{\text{link}} \cdot \log(|i - j| + \epsilon)$, and the triangle probability is   $p({\text{triangle}}) = - C_{\text{triangle}} \cdot (\log(|i - j| + \epsilon) + \log(|i - k| + \epsilon) + \log(|j - k| + \epsilon))$ where $i,j,k$ are the node positions, and $C_{\text{link}} = 3e-2$, $C_{\text{triangle}} = 5e-5$, $\epsilon=1e-4$.\\
\end{itemize}

The objective was to empirically demonstrate the consistency of the spectral norm bounds and centrality measures (degree, eigenvector, and Katz centrality) with our theoretical findings.

Our approach involved generating random graphs of varying sizes and complexities. For each graph, we computed the norm distance between the simulated and expected adjacency matrices, along with various centrality measures. The simulations were repeated five times to ensure statistical robustness, averaging the errors across trials.

\begin{figure}[!h]
\centering
\includegraphics[width=3.5in]{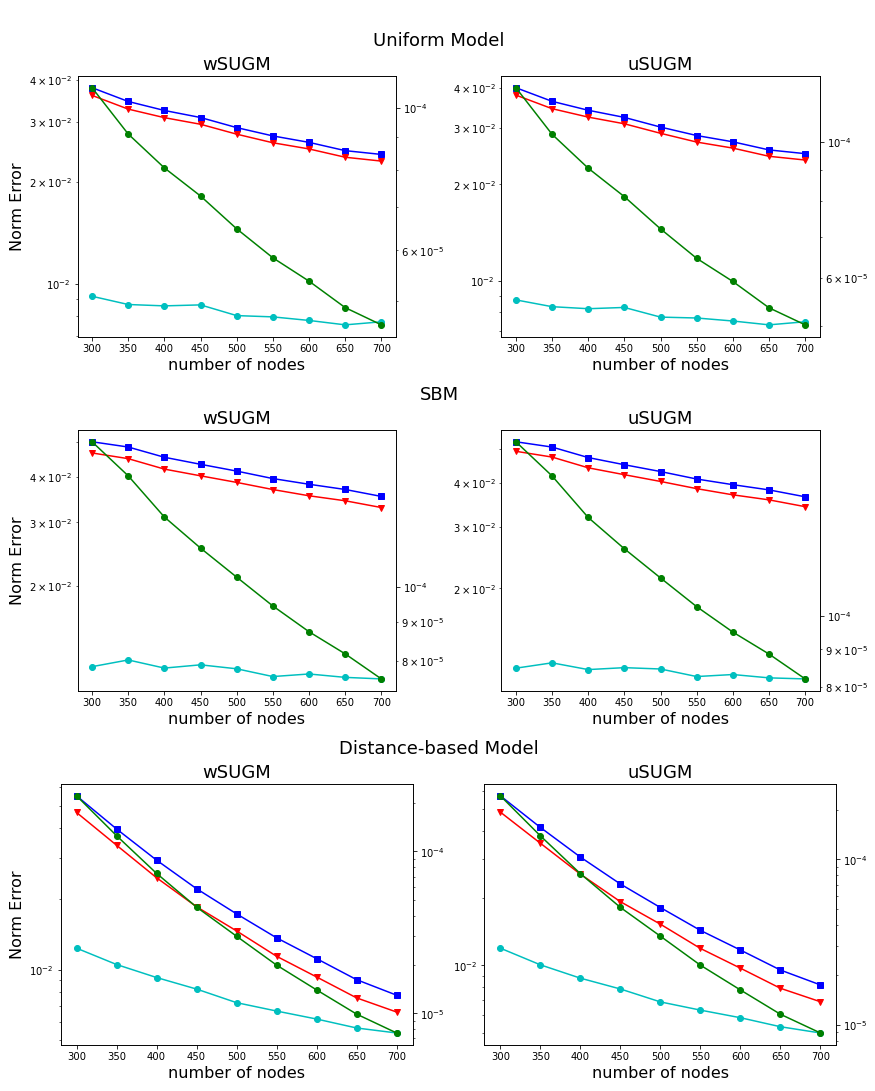}
\caption{Concentration results for Subgraph Generated Models (SUGMs) across network sizes: \\(1) Norm (navy): $\norm{\Bar{A} - \Exp{\Bar{A}}}_2$, left y-axis; \\(2) Degree (red): $\norm{c^d_{\alpha} (\Bar{A}) - c^d_{\alpha} (\Exp{\Bar{A}})}_1 / n$, left y-axis; \\(3) Eigenvector (cyan): $\norm{c^e_{\alpha} (\Bar{A}) - c^e_{\alpha} (\Exp{\Bar{A}})}_1/ n$, left y-axis; \\(4) Katz (green): $\norm{c^k_{\alpha} (\Bar{A}) - c^k_{\alpha} (\Exp{\Bar{A}})}_1/ n$, right y-axis.\\Plots report the average of these quantities over 5 random network realizations for each network size. }
\label{fig:error}
\end{figure}

The above plots display the error metrics for the weighted and unweighted SUGMs with the three connection models. In all cases, as the number of nodes increases, the error in the norm and centrality measures generally decreases, in log-log scale. These trends substantiate our theoretical results in the previous sections, and is critical for practical applications in large-scale network analysis.

\section{Discussion and Future Work}

In this article, we derived a probability bound for the spectral norm of the difference between the realized and expected adjacency matrix of a network generated from a SUGM. As a corollary, we showed that degree, eigenvector and Katz centrality measures of nodes in networks sampled from the SUGMs  converge on average to the centrality measures of the corresponding nodes in the expected networks with high probability. As future directions we aim at proving convergence for additional centrality measures as well as for other important network models, such as network games or contagion processes.

\appendices
\section{Notation}
In the following proofs, we use the common complexity notations: $f(n) \in O(g(n))$ to denote $|f|$ is bounded above by $g$ (up to constant factor) asymptotically; $f(n) \in \Omega(g(n))$ to denote $f$ is bounded below by $g$ asymptotically; $f(n) = \Theta(g(n))$, or $f(n) \asymp g(n)$ to denote $f$ is bounded both above and below by $g$ asymptotically. 

\section{Omitted Proofs} 
\subsection{Proof of Proposition \ref{prop:main_wSUGM_easy}}\label{appendix:main_wSUGM_easy}

\begin{proof}

    Recall that 
    $$A_w \sim W(\Vec{x}) = \sum_{\substack{ t \in T \\ L \in \mathfrak{S}(n, m_t)}} X(t,L)$$
where we defined the random symmetric matrices $X(t,L) := x(t,L) A(t,L)$. 
For every $t \in T$, $L \in \mathfrak{S}(n, m_t)$,
    \begin{align*}
        \norm{X(t,L) - \Exp{X(t,L)}}_2 &= \norm{(x(t,L) - p(t,L)) A(t,L)}_2\\
        &= \abs{x(t,L) - p(t,L)} \cdot \norm{A(t,L)}_2 \\
        &\leq 1 \cdot \norm{A(t,L)}_2 .
    \end{align*}
Since $\norm{A}_2 \leq \sqrt{\norm{A}_1 \norm{A}_\infty}$ holds for any matrix $A$, we obtain 
\begin{align*}
    \norm{A(t,L)}_2 &\leq \sqrt{\norm{A(t,L)}_1 \norm{A(t,L)}_\infty} \\
    &\leq \sqrt{(m_t-1)(m_t-1)} = m_t-1.
\end{align*}
Therefore,
\begin{equation} \label{Equation:P1_E0}
    \begin{aligned}
        \norm{X(t,L) - \Exp{X(t,L)}}_2 &\leq m_t - 1 < M 
    \end{aligned}
\end{equation}  
for every $t \in T$, $L \in \mathfrak{S}(n, m_t)$. On the other hand,
\begin{align*}
    \var{X(t,L)} &= \Exp{(X(t,L) - \Exp{X(t,L)})^2}\\
    &= \Exp{(x(t,L) - p(t,L))^2 A(t,L)^2}\\
    &= \var{x(t,L)} \cdot A(t,L)^2 .
\end{align*}
We use the matrix $P_2(t,L) \in \RR^{n \times n}$ to denote $A(t,L)^2$ (i.e.,  the matrix whose entries correspond to the number of two-hop paths between two nodes in the adjacency matrix $A{(t,L)}$). Then
   \begin{align*}
        v^2 &:= \norm{\smashoperator[r]{\sum_{\substack{ t \in T \\ L \in \mathfrak{S}(n, m_t)}}} ~\var{X(t,L)} }_2 \\
        &= \norm{\smashoperator[r]{\sum_{\substack{ t \in T \\ L \in \mathfrak{S}(n, m_t)}}} ~p(t,L)(1-p(t,L)) P_2(t,L) }_2 .
    \end{align*}

    Note that the entire matrix inside the spectral norm is symmetric, and for a symmetric matrix $A$,
    $\norm{A}_2 \leq \sqrt{\norm{A}_1 \norm{A}_{\infty}} = \sqrt{\norm{A}_{\infty}^2} = \norm{A}_\infty .$
We can therefore bound
    \begin{align*}
        v^2 &\leq \norm{\smashoperator[r]{\sum_{\substack{ t \in T \\ L \in \mathfrak{S}(n, m_t)}}} ~p(t,L)(1-p(t,L))  P_2(t,L) }_\infty \\
        &\leq \norm{\smashoperator[r]{\sum_{\substack{ t \in T \\ L \in \mathfrak{S}(n, m_t)}}} ~p(t,L)  P_2(t,L) }_\infty \\
        &= \max_{1 \leq k \leq n} \left\{\smashoperator[r]{\sum_{\substack{ t \in T \\ L \in \mathfrak{S}(n, m_t)}}} ~p(t,L) \left[ \sum_{1 \leq j \leq n}  P_2(t,L)_{kj}  \right] \right\} .
    \end{align*}
We use $d(t,L)_i$ to denote the degree of node $i$ in subgraph $g_t(L)$, where $d(t,L)_i = 0$ if $i \notin L$. Then for every $t \in T$, $L \in \mathfrak{S}(n, m_t)$,
    \begin{align*}
        \begin{cases}
            ~{P_2(t,L)}_{kj}  \leq d(t,L)_k \quad &\text{ if } j \in g_t(L), \\
            ~{P_2(t,L)}_{kj} =0 \quad &\text{ if } j \notin g_t(L). 
        \end{cases}
    \end{align*}
    And since $|g_t(L)| = m_t$, 
    \begin{align*}
        v^2  
        &\leq \max_{1 \leq k \leq n} \left[ \smashoperator[r]{\sum_{\substack{ t \in T \\ L \in \mathfrak{S}(n, m_t)}}} ~p(t,L) \left( m_t  \cdot d(t,L)_k\right)\right] \\
        &\leq M \cdot \max_{1 \leq k \leq n} \left[ \smashoperator[r]{\sum_{\substack{ t \in T \\ L \in \mathfrak{S}(n, m_t)}}} ~p(t,L) d(t,L)_k \right].
    \end{align*}
On the other hand, given 
    \begin{align*}
        \Delta_w &:= \norm{\Exp{A_w}}_{\infty} 
       = \max_{1 \leq k \leq n} \left[ \smashoperator[r]{\sum_{\substack{ t \in T \\ L \in \mathfrak{S}(n, m_t)}}} ~p(t,L) \cdot d(t,L)_k \right] ,
    \end{align*}
    we have
    \begin{equation} \label{Equation:P1_E1}
    \begin{aligned}
        v^2 \leq M \Delta_w < M^2 \Delta_w . 
    \end{aligned}
    \end{equation}
    Set $a = \sqrt{4M^2 \Delta_w\ln{(2n/\epsilon)}}$, so that
    \begin{align*}
        a &= \sqrt{4M^2 \Delta_w\ln{(2n/\epsilon)}} \\
        &= 3 \sqrt{M^2 \Delta_w\cdot \frac{4}{9} \ln{(2n/\epsilon)}} \\
        &< 3 \sqrt{M^2\Delta_w^2} \\
        &= 3 M\Delta_w,
    \end{align*}
    where the inequality comes from  Assumption \eqref{Assumption:P1_A1}. Therefore 
    \begin{equation} \label{Equation:P1_E2}
    \begin{aligned}
        \frac{a}{3} < M \Delta_w . 
    \end{aligned}
    \end{equation}

    Combining equations \eqref{Equation:P1_E0}, \eqref{Equation:P1_E1} and \eqref{Equation:P1_E2} with Theorem \ref{thm:chung_ineq}, we conclude that 
    \begin{align*}
        &~~~~~ \Prob{\norm{A_w - \Exp{A_w}}_2 > a} \\
        &\leq 2n \exp{\left( -\frac{a^2}{2v^2+2Ka/3}  \right)} \\
        &< 2n \exp{\left( -\frac{4M^2 \Delta_w \ln{(2n/\epsilon)}}{2 \cdot M^2 \Delta_w +2 \cdot M \cdot M\Delta_w}  \right)} \\
        &= \epsilon ~,
    \end{align*}
    which finishes the proof.
\end{proof}


\subsection{Proof of Lemma \ref{lem:VRVS_entry}} 

\begin{proof} 
    For any random vector $\Vec{x} \in \mathcal{Z}$ of length $K \in \NN$, and every $t \in T$, $L \in \mathfrak{S}(n, m_t)$, construct 
    \begin{align*}
        \dot{\Vec{x}}^{(t,L)} &= (x_1, x_2, \ldots, 1, \ldots, x_K) \\
        \ddot{\Vec{x}}^{(t,L)} &= (x_1, x_2, \ldots, 0, \ldots, x_K)
    \end{align*}
    where we enforce the result of the Bernoulli random variable entry at index $(t,L)$. Define the random matrix 
    \begin{align*}
        \mathcal{C}_U^{(t,L)}(\Vec{x}) &= U(\dot{\Vec{x}}^{(t,L)}) - U(\ddot{\Vec{x}}^{(t,L)})
    \end{align*}
    so that    
    \begin{equation} \label{step3}
    \begin{aligned}
        &~~~~~\Etl{(\hat{U}-\hat{U}^{(t,L)})^2 | \Vec{x}} \\
        &= \Etl{ \left(x(t,L) - x'(t,L)\right)^2 \cdot \mathcal{C}_U^{(t,L)}(\Vec{x})^2 | \Vec{x}}.
    \end{aligned}
    \end{equation}
    To derive \eqref{step3}  note  that if $x(t,L) = x'(t,L)$ (i.e., they are both ones or both zeros) then $\hat{U}=\hat{U}^{(t,L)}$; if instead $x(t,L) \neq x'(t,L)$ then $(\hat{U}-\hat{U}^{(t,L)})^2= (U(\dot{\Vec{x}}^{(t,L)}) - U(\ddot{\Vec{x}}^{(t,L)}))^2= \mathcal{C}_U^{(t,L)}(\Vec{x})^2.$

    From \eqref{step3} it is immediate to conclude that $0 \le_e V_U(\Vec{x})$, and we therefore focus on proving $V_U(\Vec{x})\le_e V_W(\Vec{x}).$

    The matrix $\mathcal{C}_U^{(t,L)}(\Vec{x})$ is a $0$-$1$ adjacency matrix that can only have nonzero entries on edges within $g_t(L)$. Hence
    \begin{align*}
        & ~~~~~ \Etl{(\hat{W}-\hat{W}^{(t,L)})^2 | \Vec{x}} - \Etl{(\hat{U}-\hat{U}^{(t,L)})^2 | \Vec{x}} \\
        &= \Etl{(\hat{W}-\hat{W}^{(t,L)})^2 - (\hat{U}-\hat{U}^{(t,L)})^2| \Vec{x}} \\
        &= \mathbb{E}_{{x}'(t,L)} \left[ \left(x(t,L) - x'(t,L)\right)^2 \cdot A(t,L)^2 \right.\\
        &\quad\quad \left. - \left(x(t,L) - x'(t,L)\right)^2 \cdot \mathcal{C}_U^{(t,L)}(\Vec{x})^2 | \Vec{x} \right] \\
        &= \mathbb{E}_{{x}'(t,L)} \left[ \left(x(t,L) - x'(t,L)\right)^2 \cdot (A(t,L)^2  - \mathcal{C}_U^{(t,L)}(\Vec{x})^2) | \Vec{x} \right] \\
        &\geq_{\text{e}} 0
    \end{align*}
    where the last inequality comes from the fact that $A(t,L) \geq_{\text{e}} \mathcal{C}_U^{(t,L)}(\Vec{x}) \geq_{\text{e}} 0$ implies $A(t,L)^2 \geq_{\text{e}} \mathcal{C}_U^{(t,L)}(\Vec{x})^2$. Therefore,
    \begin{align*}
         V_U(\Vec{x}) &= \frac{1}{2} \cdot \smashoperator{\sum_{\substack{ t \in T \\ L \in \mathfrak{S}(n, m_t)}}} ~\Etl{(\hat{U}-\hat{U}^{(t,L)})^2 | \Vec{x}} \\
         &\leq_{\text{e}} \frac{1}{2} \cdot  \smashoperator{\sum_{\substack{ t \in T \\ L \in \mathfrak{S}(n, m_t)}}} ~\Etl{(\hat{W}-\hat{W}^{(t,L)})^2 | \Vec{x}} \\
         &= V_W(\Vec{x}).
    \end{align*}
\end{proof}


\subsection{Proof of Lemma \ref{lem:VRVS_overall}} 

\begin{proof} 
    From Lemma \ref{lem:VRVS_entry}, for any $\Vec{x} \in \mathcal{Z}$,
    $$0 \leq_{\text{e}} V_U(\Vec{x}) \leq_{\text{e}} V_W(\Vec{x}).$$
 
    Then for any $\psi > 0$, and integer $k > 0$,
    $$0 \leq_{\text{e}} (\psi V_U(\Vec{x}))^k \leq_{\text{e}} (\psi V_W(\Vec{x}))^k.$$

    By the definition of matrix exponential, for any $\Vec{x} \in \mathcal{Z}$,
    \begin{align*}
        \exp (\psi V_U(\Vec{x})) &\leq_{\text{e}} \exp (\psi V_W(\Vec{x})) \\
        \Longrightarrow \atrace \exp (\psi V_U(\Vec{x})) &\leq \atrace \exp (\psi V_W(\Vec{x})) \\
        \log \Exp{\atrace \exp (\psi V_U)} &\leq \log \Exp{\atrace \exp (\psi V_W)} .
    \end{align*}    
\end{proof}


\subsection{Proof of Lemma \ref{lem:bound_weight_proxy}} \label{appendix:bound_weight_proxy}  

We list several useful lemmas here before presenting the proof of Lemma \ref{lem:bound_weight_proxy}.

\begin{lemma} \label{lem:LogCon} [Section 2.5, \cite{tropp2012user}]
    The matrix logarithm operator is monotone
    \begin{align*}
        0 \prec A \preceq B \Longrightarrow \log A \preceq \log B
    \end{align*}
    and concave
    \begin{align*}
        \alpha \log A + (1-\alpha) \log B \preceq \log (\alpha A + (1-\alpha) B)
    \end{align*}
    where $A,B$ are positive definite matrices and $\alpha \in [0,1]$.
\end{lemma}

\begin{lemma} \label{lem:ChernoffMGF} (Lemma 5.8, \cite{tropp2012user})
    Suppose $X$ is a random positive semidefinite matrix that satisfies $\lambda_{\max}(X) \leq 1$. Then, for $\theta \in \RR$,
    \begin{align*}
        \Exp{e^{\theta X}} \preceq I + (e^\theta - 1) \Exp{X}.
    \end{align*}
\end{lemma}

\begin{lemma} \label{lem:independent}
    Define the random matrix $D_W^{(t,L)}(\Vec{x}) := \Etl{(\hat{W}-\hat{W}^{(t,L)})^2 | \Vec{x}}$ for $\hat W,\hat{W}^{(t,L)}$ as in \eqref{W}. The sequence $ \{ D_W^{(t,L)}(\Vec{x}) \}_{t \in T, L \in \mathfrak{S}(n, m_t)}$ is a finite sequence of independent, random, symmetric matrices.
\end{lemma}

\begin{proof}
    To show that the terms are independent, for a fixed $\Vec{x} \in \mathcal{Z}$, fix $t_0 \in T$ and $L_0 \in \mathfrak{S}(n, m_{t_0})$, we have
    \begin{equation}\label{step4}
    \begin{aligned}
        &~~~~~D_W^{(t_0,L_0)}(\Vec{x}) \\
        &= \Exp{(\hat{W}-\hat{W}^{(t_0,L_0)})^2 | \Vec{x}} \\
        &= \Exp{((W(\Vec{x}) \!- \! \Exp{W(\Vec{x})})\!-\!(W(\Vec{x}^{(t,L)}) \!- \!\Exp{W(\Vec{x}^{(t,L)})}))^2 | \Vec{x}} \\
        &= \Exp{\left(W(\Vec{x}) - W(\Vec{x}^{(t,L)}) \right)^2 | \Vec{x}} \\
        &= \Exp{\left(x(t_0,L_0) A(t_0,L_0) - x'(t_0,L_0) A(t_0,L_0)\right)^2 | \Vec{x}} \\
        &= \Exp{\left(x(t_0,L_0) - x'(t_0,L_0)\right)^2 | \Vec{x}} \cdot A(t_0,L_0)^2 \\
        &= \Exp{\left(x(t_0,L_0) - x'(t_0,L_0)\right)^2 | x(t_0,L_0)} \cdot A(t_0,L_0)^2 
    \end{aligned}
    \end{equation}
    which only depends on the Bernoulli random variable $x(t_0,L_0)$. In other words, $D_W^{(t_0,L_0)}(\Vec{x})=D_W^{(t_0,L_0)}(x(t_0,L_0))$.
\end{proof}

\begin{lemma} \label{lem:spectrum}
    For any $\Vec{x} \in \mathcal{Z}$, $t \in T$, and $L \in \mathfrak{S}(n, m_t)$, $D_W^{(t,L)}(\Vec{x})$ is a positive semi-definite matrix and
    \begin{align*}
        \norm{D_W^{(t,L)}(\Vec{x})}_2 < 2 M^2 ~\textit{ almost surely}
    \end{align*}
    where $M = \max_{t \in T} \{m_t\}$.
\end{lemma}

\begin{proof}
    The semi-definite positiveness is obvious from the quadratic form in  the last line of \eqref{step4} in Lemma \ref{lem:independent}. For the upper bound, first recall from \eqref{step4} in Lemma \ref{lem:independent} that
    \begin{align*}
    &~~~~~D_W^{(t,L)}(\Vec{x}) = D_W^{(t,L)}({x}(t,L)) \\
    &= \Exp{\left(x(t,L) - x'(t,L)\right)^2 | x(t,L)} \cdot A(t,L)^2.
    \end{align*}
    Hence w.p.1,
    \begin{align*}
        \norm{D_W^{(t,L)}(\Vec{x})}_2 & =\norm{\Exp{\left(x(t,L) - x'(t,L)\right)^2 | x(t,L)} \cdot A(t,L)^2}_2\\
        &=\Exp{\left(x(t,L) - x'(t,L)\right)^2 | x(t,L)} \norm{  A(t,L)^2}_2\\
        &< 2 \norm{A(t,L)^2}_2 \leq 2 \norm{A(t,L)}_2^2 \\
        &\leq 2 \norm{A(t,L)}_\infty \norm{A(t,L)}_1 \\
        &\leq 2 (m_t-1) (m_t - 1) \\
        &< 2 M^2.
    \end{align*}
\end{proof}

\begin{lemma} \label{lem:subaddCGF} (Lemma 3.4, \cite{tropp2012user})
    For a finite sequence $\{X_k\}$ of independent, random, symmetric matrices, and any $\theta > 0$,
    \begin{align*}
        \Exp{ \trace \exp{ \left( \sum_{k} \theta X_k \right)}} \leq \trace \exp \sum_{k} \log \Exp{\exp \left( \theta X_k \right)}.
    \end{align*}
\end{lemma}

\begin{proof} [\textbf{Proof of Lemma \ref{lem:bound_weight_proxy}}]
For $\psi > 0$, we have,

\begin{equation}\label{Equation:bound_E1}
\begin{aligned} 
    &~~~~~\Exp{\trace \exp (\psi V_W)} \\
    &= \Exp{\trace \exp \left(\frac{\psi}{2} \cdot \smashoperator{\sum_{\substack{ t \in T \\ L \in \mathfrak{S}(n, m_t)}}} ~D_W^{(t,L)}(\Vec{x}) \right)} \\
    &\leq \trace \exp{ \left( \smashoperator[r]{\sum_{\substack{ t \in T \\ L \in \mathfrak{S}(n, m_t)}}} ~\log \Exp{\exp \left( \frac{\psi}{2}D_W^{(t,L)}(\Vec{x}) \right) }    \right) }
\end{aligned}
\end{equation}
where the inequality on the second line comes from Lemma \ref{lem:independent} and \ref{lem:subaddCGF}. 

Examine the inner summation part from above, letting $K \in \NN$ be the total number of terms in the summation, 
\begin{align*}
    \mathcal{A} &:= \smashoperator[r]{\sum_{\substack{ t \in T \\ L \in \mathfrak{S}(n, m_t)}}} ~\log \Exp{\exp \left( \frac{\psi}{2}D_W^{(t,L)}(\Vec{x})\right)} \\
    &= K \cdot \frac{1}{K} \cdot \smashoperator{\sum_{\substack{ t \in T \\ L \in \mathfrak{S}(n, m_t)}}} ~\log \Exp{\exp \left( \frac{\psi}{2}D_W^{(t,L)}(\Vec{x})\right)} \\
    &\preceq K \cdot \log \left( \frac{1}{K} \cdot \smashoperator{\sum_{\substack{ t \in T \\ L \in \mathfrak{S}(n, m_t)}}} ~\Exp{\exp \left( \frac{\psi}{2}D_W^{(t,L)}(\Vec{x})\right)} \right)
\end{align*}
where the last line uses the concavity property of the matrix logarithm from Lemma \ref{lem:LogCon} and the fact that matrix exponential $\exp \left( \frac{\psi}{2}D_W^{(t,L)}(\Vec{x})\right)$ is positive definite. 

Now, apply Lemma \ref{lem:ChernoffMGF} with $X = \frac{\psi}{2}D_W^{(t,L)}(\Vec{x})$, $\theta = 1$. By Lemma \ref{lem:spectrum}, $X$ is positive semi-definite and $\lambda_{\max}(X) = \frac{\psi}{2} \norm{D_W^{(t,L)}(\Vec{x})}_2 < 1$ a.s. , since $\psi<1/M^2$, hence
\begin{align*}
    \Exp{\exp \left( \frac{\psi}{2}D_W^{(t,L)}(\Vec{x})\right)} &\preceq I + (e-1) \Exp{\frac{\psi}{2}D_W^{(t,L)}(\Vec{x})} \\
    &= I + \frac{(e-1)\psi}{2} \cdot \Exp{D_W^{(t,L)}(\Vec{x})}.
\end{align*}
Therefore, the summation 
\begin{equation}\label{step5}
\begin{aligned}
    \mathcal{A} &\preceq K \cdot \log \left( \frac{1}{K}  \cdot \smashoperator{\sum_{\substack{ t \in T \\ L \in \mathfrak{S}(n, m_t)}}} ~\Exp{\exp \left( \frac{\psi}{2}D_W^{(t,L)}(\Vec{x})\right)} \right) \\
    &\preceq K \cdot \log \left( \frac{1}{K} \cdot \smashoperator{\sum_{\substack{ t \in T \\ L \in \mathfrak{S}(n, m_t)}}} ~\left( I + \frac{(e-1)\psi}{2} \cdot \Exp{D_W^{(t,L)}(\Vec{x})} \right) \right) \\
    &= K \cdot \log \left( I + \frac{(e-1)\psi}{2K} \cdot \smashoperator{\sum_{\substack{ t \in T \\ L \in \mathfrak{S}(n, m_t)}}} ~\Exp{D_W^{(t,L)}(\Vec{x})} \right) =: \mathcal{B}
\end{aligned}
\end{equation}
where the second to last line uses the monotonicity of the matrix logarithm operator from Lemma~\ref{lem:LogCon}. 

Since the trace exponential function is monotone with respect to the semidefinite order (i.e., 
$
    \mathcal{A} \preceq \mathcal{B} \Longrightarrow \trace e^{\mathcal{A}} \leq \trace e^{\mathcal{B}}
$, Sec. 2, \cite{petz1994survey})
we have from equations \eqref{Equation:bound_E1} and \eqref{step5} that, 
\begin{equation}\label{Equation:bound_E2}
\begin{aligned}
    &~~~~~ \Exp{\trace \exp (\psi V_W)} \\
    &\leq \trace \exp{ \mathcal{A}} 
    \leq \trace \exp{\mathcal{B}} \\
    &\leq n \cdot \lambda_{\max} (\exp {\mathcal{B}}) \\
    &= n \!\cdot\! \exp{ \!\left\{ K \!\cdot\! \log \lambda_{\max} \left( I + \frac{(e-1)\psi}{2K} \!\cdot\! \smashoperator{\sum_{\substack{ t \in T \\ L \in \mathfrak{S}(n, m_t)}}} ~\Exp{D_W^{(t,L)}(\Vec{x})} \right) \!\right\} }
\end{aligned}
\end{equation}

where the last equality comes from applying the spectral mapping theorem onto the matrix exponential and the matrix logarithm. Note that
\begin{equation} \label{Equation:bound_E3}
\begin{aligned}
    &~~~~~ \lambda_{\max} \left( I + \frac{(e-1)\psi}{2K} \cdot \smashoperator{\sum_{\substack{ t \in T \\ L \in \mathfrak{S}(n, m_t)}}} ~\Exp{D_W^{(t,L)}(\Vec{x})} \right) 
    \\
    &= 1 + \frac{(e-1)\psi}{2K} \norm{\smashoperator[r]{\sum_{\substack{ t \in T \\ L \in \mathfrak{S}(n, m_t)}}} ~\Exp{D_W^{(t,L)}(\Vec{x})}}_2 ,
\end{aligned}
\end{equation}
where we used $\lambda_{\max} = \norm{A}_2$ for symmetric positive semidefinite matrix $A$.

From Lemma \ref{lem:independent}, we have calculated
\begin{align*}
    &~~~~~ D_W^{(t,L)}(\Vec{x}) \\
    &= \Exp{(x(t,L) - x'(t,L))^2 | x(t,L)} \cdot A(t,L)^2 \\
    &= \left( x(t,L)^2 - 2 x(t,L) p(t,L) + \Exp{x'(t,L)^2} \right) \!\cdot\! A(t,L)^2 \\
    \Longrightarrow  &~~~~~\Exp{D_W^{(t,L)}(\Vec{x})} \\
    &= \left(\Exp{x(t,L)^2} \!-\! 2 p(t,L) p(t,L) \!+\! \Exp{x'(t,L)^2} \right) \!\cdot\! A(t,L)^2 \\
    &= 2 \left( \Exp{x(t,L)^2} - \Exp{x(t,L)}^2 \right) \cdot A(t,L)^2 \\
    &= 2  \var{x(t,L)} \cdot A(t,L)^2.
\end{align*}

Substituting in \eqref{Equation:bound_E3},
\begin{align*}
    &~~~~~\lambda_{\max} \left( I + \frac{(e-1)\psi}{2K} \cdot \smashoperator{\sum_{\substack{ t \in T \\ L \in \mathfrak{S}(n, m_t)}}} ~\Exp{D_W^{(t,L)}(\Vec{x})} \right) \\
    &= 1 + \frac{(e-1)\psi}{2K} \norm{\smashoperator[r]{\sum_{\substack{ t \in T \\ L \in \mathfrak{S}(n, m_t)}}} ~\Exp{D_W^{(t,L)}(\Vec{x})}}_2 \\
    &= 1 + \frac{(e-1)\psi}{K} \norm{\smashoperator[r]{\sum_{\substack{ t \in T \\ L \in \mathfrak{S}(n, m_t)}}} ~\var{x(t,L)} \cdot A(t,L)^2 }_2 \\
    &= 1 + \frac{(e-1)\psi}{K} \cdot v^2 \\
    &\leq  1 + \frac{(e-1)\psi}{K} \cdot \Delta_w M
\end{align*}
where we used $v^2 \leq \Delta_w M$ as derived in the proof of Proposition~\ref{prop:main_wSUGM_easy} before. 

Substituting into \eqref{Equation:bound_E2},
\begin{align*}
    &~~~~~\Exp{\trace \exp (\psi V_W)}  \\
    &\leq n \cdot \exp{ \left( K \cdot \log \left( 1 + \frac{(e-1)\psi}{K} \cdot \Delta_w M \right) \right) }.
\end{align*}

Finally, since $\log(1+a) \leq a $ for $a > 0$,
\begin{align*}
    \Exp{\trace \exp (\psi V_W)} &\leq n \cdot \exp{ \left\{ K \cdot \frac{(e-1)\Delta_w M }{K} \cdot \psi \right\} } \\
    &= n \cdot \exp{ \left\{ (e-1)\Delta_w M  \cdot \psi \right\} } \\
    \Longrightarrow \Exp{\atrace \exp (\psi V_W)} &\leq \exp{ \left\{ (e-1)\Delta_w M  \cdot \psi \right\} } \\
    \log \Exp{\atrace \exp (\psi V_W)} &\leq (e-1)\Delta_w M  \cdot \psi
\end{align*}
\end{proof}


\subsection{Proof of Proposition \ref{prop:main_uSUGM_ES}} \label{appendix:main_uSUGM_ES}

\begin{proof}
    By Theorem \ref{thm:es} recalled in the main text, when $|\theta| < \sqrt{\psi/2}$, 
    \begin{align*}
        \log m_{\hat{U}(\Vec{x})}(\theta) = \log \Exp{\atrace \exp (\theta \cdot \hat{U}(\Vec{x}))} \\
        \leq \frac{\theta^2 / \psi}{1 - 2\theta^2 / \psi} \log \Exp{\atrace \exp (\psi V_U)}.
    \end{align*}
    
    From Lemma \ref{lem:VRVS_overall} and \ref{lem:bound_weight_proxy}, for $0 < \psi < \frac{1}{M^2}$, 
    \begin{align*}
        \log \Exp{ \atrace e^{\psi V_U}} \leq \log \Exp{ \atrace e^{\psi V_W}} \leq (e-1)\Delta_w M  \cdot \psi.
    \end{align*}

    Hence for $n$ sufficiently large, given Assumption \eqref{Assumption:P2_A2}, 
    \begin{align*}
        \log \Exp{ \atrace e^{\psi V_U}} 
        &\leq  (e-1) \mu_{T,p} \Delta_u M  \cdot \psi.
    \end{align*}
    
    Denote the scalar $\chi = (e-1) \mu_{T,p} \Delta_u M $, then
    \begin{align*}
        \log m_{\hat{U}(\Vec{x})}(\theta) &\leq \frac{\theta^2 / \psi}{1 - 2\theta^2 / \psi} \cdot \chi \psi \\
        &= \frac{\theta^2 \chi}{1 - 2\theta^2 / \psi} .
    \end{align*}

    Fix  $\psi = \frac{1}{2M^2} < 1$. By Theorem \ref{thm:matrix_lap} in the main text, for $t \in \RR_{>0}$
    \begin{align*}
        \Prob{\lambda_{\max}(\hat{U}(\Vec{x})) \geq t } &\leq n \cdot \inf_{\theta > 0} \exp \left\{ -\theta t + \log m_{\hat{U}(\Vec{x})}(\theta) \right\} \\
        &\leq n \cdot \smashoperator[lr]{\inf_{ 0 < \theta < \sqrt{\psi/2}}} ~\exp \left\{ -\theta t + \log m_{\hat{U}(\Vec{x})}(\theta) \right\} \\
        &\leq n \cdot \smashoperator[lr]{\inf_{ 0 < \theta < \sqrt{\psi/2}}} ~\exp \left\{ -\theta t + \frac{\theta^2 \chi}{1 - 2\theta^2 / \psi} \right\}. 
    \end{align*}
 We follow arguments similar to the ones in Section 4.2.4 of \cite{mackey2014matrix} to bound the term in the right hand side of the previous inequality. Specifically, 
    set  $c = \frac{1}{2\psi} = M^2$ and  $\theta^* = \frac{1-1/\sqrt{1+4ct/\chi}}{4c}$. Then
    \begin{align*}
        (1)~~~  \theta^* &= \frac{1-1/\sqrt{1+4ct/\chi}}{4c} > 0, \\
        (2)~~~ \theta^* &= \frac{1-1/\sqrt{1+4ct/\chi}}{4c} < \frac{1}{4c} =\frac{\psi}{2}< \sqrt{\frac{\psi}{2}}.
    \end{align*}
    
    Hence $\theta^* \in (0, \sqrt{\psi/2})$. Using 
    \begin{align*}
        {\theta^*}^2 / \psi  = {\theta^*} \cdot {\theta^*} / \psi < 1 &\Longrightarrow 1 - {\theta^*}^2 / \psi > 0 \\
        c {\theta^*} < \frac{1}{4} &\Longrightarrow 1 - 4c {\theta^*} > 0 \\
        2{\theta^*} / \psi < 1 < 4c &\Longrightarrow 1 - 2{\theta^*}^2 / \psi > 1 - 4c {\theta^*} ,
    \end{align*}
    we have
    \begin{align*}
        &~~~~~\Prob{\lambda_{\max}(\hat{U}(\Vec{x})) \geq t } \\
        &\leq n \cdot \exp \left\{ -{\theta^*} t + \frac{{\theta^*}^2 \chi}{1 - 2{\theta^*}^2 / \psi} \right\}\\
        &\leq n \cdot \exp \left\{ -\theta^* t + \frac{{\theta^*}^2 \chi}{1 - 4c{\theta^*}} \right\}\\
        &= n \cdot \exp \left\{ - \frac{\chi}{16c^2} (1- \sqrt{1+4ct/\chi})^2 \right\} \\
        &\leq n \cdot \exp \left\{ -\frac{t^2}{4\chi + 8ct} \right\}
    \end{align*}
    where the last line depends on the numerical fact that
    \begin{align*}
        (1- \sqrt{1+2x})^2 \geq \frac{x^2}{1+x} ~\text{ for all }~ x \geq 0.
    \end{align*}

    Set $t = \sqrt{16 M^4 \mu_{T,p} \Delta_u  \ln{(2n/\epsilon)}}$ so that  
    \begin{align*}
        t &=  \sqrt{ M^4 \mu_{T,p} \Delta_u  \cdot 16 \ln{(2n/\epsilon)}} \\
        &< \sqrt{M^4\mu_{T,p} \Delta_u  \cdot \mu_{T,p} \Delta_u } \\
        &= M^2\mu_{T,p} \Delta_u ,
    \end{align*}
    where the inequality comes from the Assumption \eqref{Assumption:P2_A3}. For this choice of $t$
    \begin{align*}
        &~~~~~\Prob{\lambda_{\max}(\hat{U}(\Vec{x})) \geq t }  \\
        &\leq n \cdot \exp \left\{ -\frac{t^2}{4 \chi + 8ct} \right\} \\
        &= n \cdot \exp \left\{ -\frac{ 16 M^4 \mu_{T,p} \Delta_u  \ln{(2n/\epsilon)} }{4 \cdot (e-1)\mu_{T,p} \Delta_u  M + 8 \cdot M^2 \cdot t} \right\} \\
        &< n \cdot \exp \left\{ -\frac{ 16 M^4 \mu_{T,p} \Delta_u \ln{(2n/\epsilon)} }{8\mu_{T,p} \Delta_u  M^4 + 8M^2 \cdot M^2\mu_{T,p} \Delta_u } \right\} \\
        &= n \cdot \exp \left\{ -\frac{ 16 M^4 \mu_{T,p} \Delta_u  \ln{(2n/\epsilon)} }{16\mu_{T,p} \Delta_u  M^4 } \right\} \\
        &= \epsilon /2.
    \end{align*}

    We then apply Theorem \ref{thm:matrix_lap} again, this time on the minimum eigenvalue of $\hat{U}(\Vec{x})$. For $t \in \RR_{>0}$, 
    \begin{align*}
        &~~~~~\Prob{\lambda_{\min}(\hat{U}(\Vec{x})) \leq -t } \\
        &\leq n \cdot \inf_{\theta < 0} \exp \left\{ -\theta (-t) + \log m_{\hat{U}(\Vec{x})}(\theta) \right\} \\
        &= n \cdot \inf_{\theta > 0} \exp \left\{ -(-\theta) (-t) + \log m_{\hat{U}(\Vec{x})}(-\theta) \right\} \\
        &= n \cdot \inf_{\theta > 0} \exp \left\{ -\theta t + \log m_{\hat{U}(\Vec{x})}(-\theta) \right\} \\
        &\leq n \cdot \smashoperator[lr]{\inf_{ 0 < \theta < \sqrt{\psi/2}}} ~\exp \left\{ -\theta t + \log m_{\hat{U}(\Vec{x})}(-\theta) \right\}. 
    \end{align*}

    Note that since Theorem \ref{thm:es} works as long as $|\theta| < \sqrt{\psi/2}$, the rest of the argument follows the same as in the preceding paragraphs and we can conclude that with the same choice of $t = \sqrt{16 M^4 \mu_{T,p} \Delta_u  \ln{(2n/\epsilon)}}$,
    \begin{align*}
        \Prob{\lambda_{\min}(\hat{U}(\Vec{x})) \leq -t } < \epsilon /2.
    \end{align*}

    Finally, by union bound, with probability at least $1-\epsilon$, $\norm{A_u - \Exp{A_u}}_2 \leq 4M^2\sqrt{ \mu_{T,p} \Delta_u \ln{(2n/\epsilon)}}$.
\end{proof}


\subsection{Proof of Lemma \ref{lem:prob_frac}} \label{appendix:prob_frac}

\begin{lemma} \label{lem:delta_order}
    Let $G_w$ be the random graphs of size $n$ generated by the model wSUGM$(n,T,p)$, with probabilities $p(\cdot,\cdot)$ following Assumption \eqref{Assumption:prob_A5} and \eqref{Assumption:prob_A6}, and $A_w$ be the corresponding adjacency matrix. Denote $\Delta_w := \norm{\Exp{A_w}}_{\infty}$. There exists some $\gamma := \gamma_{T,p} \in (0,1]$, independent of $n$, such that
    $$\Delta_w \in \Theta(n^\gamma)$$
\end{lemma}

\begin{proof} 
    Using $d(t,L)_i$ to denote the degree of node $i$ in subgraph $g_t(L)$, where $d(t,L)_i = 0$ if $i \notin L$. We have
    \begin{align*}
        \Delta_w &:= \norm{\Exp{A}}_{\infty} = \max_{1 \leq k \leq n} \left[ \smashoperator[r]{\sum_{\substack{ t \in T \\ L \in \mathfrak{S}(n, m_t)}}} ~ p(t,L) \cdot d(t,L)_k \right] .
    \end{align*}

    On one hand, 
    \begin{align*}
        \Delta_w 
        &\geq \frac{1}{n} \sum_{t \in T} \left[ \smashoperator[r]{\sum_{\substack{L \in \mathfrak{S}(n, m_t)}}} ~p(t,L) \cdot \sum_{1 \leq k \leq n} d(t,L)_k \right] \\
        &\geq \frac{1}{n} \sum_{t \in T} \left[ \sum_{\substack{L \in \mathfrak{S}(n, m_t) 
        \\ p(t,L) > 0}} \frac{l_t}{n^{h_t}} \cdot \sum_{1 \leq k \leq n} d(t,L)_k \right] \\
        &\geq \frac{1}{n} \sum_{t \in T} \left[ \frac{l_t}{n^{h_t}} \cdot \sum_{\substack{L \in \mathfrak{S}(n, m_t) \\ p(t,L) > 0}} \sum_{1 \leq k \leq n} d(t,L)_k \right] \\
        &\geq \frac{1}{n} \sum_{t \in T} \left[ \frac{l_t}{n^{h_t}} \cdot 2 \xi_t {{n}\choose{m_t}} \right] \\
        & \geq  \sum_{t \in T} \frac{2 l_t}{n^{h_t+1}} \cdot \xi_t \left(\frac{n}{m_t} \right)^{m_t} \in \Theta\left(  \sum_{t \in T} n^{m_t-1- h_t} \right) .
    \end{align*}

    On the other hand, 
    \begin{align*}
        \Delta_w
        &\leq \max_{1 \leq k \leq n} \left[ \sum_{t \in T} \frac{u_t}{n^{h_t}} \cdot \left( \sum_{{L \in \mathfrak{S}(n, m_t)|k \in L}} m_t - 1 \right) \right] \\
        &= \max_{1 \leq k \leq n} \left[ \sum_{t \in T} \frac{u_t(m_t-1)}{n^{h_t}} \cdot \left( \sum_{{L \in \mathfrak{S}(n, m_t)|k \in L}} 1 \right) \right] \\
        &\leq \max_{1 \leq k \leq n} \left[ \sum_{t \in T} \frac{u_t (m_t-1)}{n^{h_t}} \cdot {{n-1}\choose{m_t-1}} \cdot m_t!\right] \\
        &\le \sum_{t \in T} \frac{u_t (m_t-1)}{n^{h_t}} \cdot {\left(\frac{e(n-1)}{m_t-1}\right)^{m_t-1}} \cdot m_t! \\
        &\in \Theta\left(  \sum_{t \in T} n^{m_t-1- h_t} \right) 
    \end{align*}
    where we used ${n \choose k} \leq \left(\frac{en}{k} \right)^k$.
    
    The conclusion follows since $m_t-1-h_t\in (0,1]$ for any $t$ by assumption.
\end{proof}

\begin{lemma} \label{cor:binomial}
    Given function $F(n) : \NN \rightarrow \RR$ where
    $$F(n) =  1-(1 - \frac{a}{n^{\alpha}})^{b(n-k)^\beta},$$
    with $k \in \NN$ fixed, $a,b,\alpha,\beta > 0$ and $\alpha > \beta$, we have
    $$F(n) \in \Theta(n^{\beta-\alpha}).$$
\end{lemma}

\begin{proof}
    Define the function 
    \begin{align*}
        G(n) &= (1 - \frac{a}{n^{\alpha}})^{b(n-k)^\beta} \\
        &= \exp{\left(  {b(n-k)^\beta} \ln{\left(1 - \frac{a}{n^{\alpha}}\right)} \right)} \\
        \Longrightarrow 
        \ln G(n) &= {b(n-k)^\beta} \ln{\left(1 - \frac{a}{n^{\alpha}}\right)}  \\
        &= {b(n-k)^\beta} \left[ - \frac{a}{n^{\alpha}} - \sum_{h = 2}^{\infty} \frac{1}{h} \left(\frac{a}{n^\alpha}\right)^h \right]
    \end{align*}
    where we apply Taylor expansion on $\ln(\cdot)$ around 1. 

    Hence,
    \begin{align*}
        G(n) &= \exp{\left( -\frac{ab(n-k)^\beta}{n^\alpha} - b(n-k)^\beta \cdot \sum_{h = 2}^{\infty} \frac{1}{h} \left(\frac{a}{n^\alpha}\right)^h  \right)} \\
        &\asymp \exp{\left( - ab \cdot n^{\beta-\alpha} - b \cdot n^\beta \cdot \sum_{h = 2}^{\infty} \frac{1}{h} \left(\frac{a}{n^\alpha}\right)^h  \right)} \\
        &= 1 - ab \cdot n^{\beta-\alpha} - b \cdot n^\beta \cdot \sum_{h = 2}^{\infty} \frac{1}{h} \left(\frac{a}{n^\alpha}\right)^h \\
        &~~~~ + \sum_{s = 2}^\infty \frac{1}{s!} \left[ - ab \cdot n^{\beta-\alpha} - b \cdot n^\beta \cdot \sum_{h = 2}^{\infty} \frac{1}{h} \left(\frac{a}{n^\alpha}\right)^h \right]^s \\
        &\asymp 1 - ab \cdot n^{\beta-\alpha}
    \end{align*}
    as we apply Taylor expansion on $\exp(\cdot)$ around 0 and use the fact that $\alpha > \beta$ to suppress higher order terms. 

    Therefore, $F(n) = 1 - G(n) \asymp ab \cdot n^{\beta - \alpha}$.
\end{proof}

\begin{proof} [\textbf{Proof of Lemma \ref{lem:prob_frac}}] 
    Given wSUGM$(n,T,p)$ and Assumption \eqref{Assumption:prob_A5} and \eqref{Assumption:prob_A6}, by Lemma \ref{lem:delta_order}, 
    $$\Delta_w \in \Theta(n^\gamma)$$
    for some $\gamma := \gamma_{T,p} \in (0,1]$, independent of $n$. Therefore, given some fixed $\epsilon > 0$, for $n$ sufficiently large,
    \begin{equation*} 
        \Delta_w > 16 \ln{(2n/\epsilon)} ,
    \end{equation*}
    and Assumption \eqref{Assumption:P1_A1} and \eqref{Assumption:P3_A4} are satisfied. 

    To show the second half of this proof, we follow the proof of Lemma \ref{lem:delta_order}. Let $t^*\in \arg\max_t [m_t-h_t-1 ]$ with minimum $m_t$, and denote $\bar \gamma = m_{t^*}-h_{t^*}-1 \in (0,1]$, which gives
    \begin{align*}
        \Delta_w &\in \Theta(n^{\bar \gamma}).
    \end{align*} 

    Now, let $A_{u|t}$ and $\Delta_{u|t}$ denote the unweighted adjacency matrix and its maximum expected degree when only subgraphs of type $t$ are generated, and let $N^t_{i,j}$ denote the number of subgraphs of type $t$ with non-zero generating probability that contains node $i$ and $j$. Clearly $\Delta_u \ge \Delta_{u|t^*}$. We distinguish two cases:

    \begin{enumerate}
        \item If $t^*=\textup{links}$, then since there is no overlap among links, $\Delta_{u|t^*} \in \Theta(n^{\bar \gamma})$ following the same derivation of the proof of Lemma \ref{lem:delta_order}.
        
        \item If $t^*\neq$ \textup{links}, then $\Exp{A_{u|t^*}}_{i,j}$ equals the probability that at least one subgraph of type $t^*$ containing node $i$ and $j$ is being generated. 
        
        On one hand, we have
        \begin{align*}
            \Delta_{u|t^*} &\leq n \cdot \max_{i,j} \left( 1-\left(1-\frac{u_{t^*}}{n^{h_{t^*}}}\right)^{N^{t^*}_{i,j}} \right) \\
            &\leq n \cdot \left( 1-\left(1-\frac{u_{t^*}}{n^{h_{t^*}}}\right)^{m_{t^*}! \cdot {{n-2} \choose {m_{t^*}-2}}} \right) \\
            &\leq n \cdot \left( 1-\left(1-\frac{u_{t^*}}{n^{h_{t^*}}}\right)^{m_{t^*}! \cdot (\frac{en}{m_{t^*}-2})^{m_{t^*}-2}} \right) \\
            &= n \cdot \left( 1 - \left(1 - \frac{a}{n^{\alpha}}\right)^{b(n-k)^\beta}\right)
        \end{align*}
        with $k = 0$, $a = u_{t^*}$, $b = m_{t^*}! \cdot (\frac{e}{m_{t^*}-2})^{m_{t^*}-2}$, $\alpha = h_{t^*}$, $\beta = m_{t^*} - 2$. Since $\alpha = h_{t^*} > m_{t^*} - 2 = \beta > 0$ from Assumption \eqref{Assumption:prob_A5}, we can apply Lemma \ref{cor:binomial} and conclude that 
        \begin{align*}
            \Delta_{u|t^*} &\leq n \cdot \left( 1 - \left(1 - \frac{a}{n^{\alpha}}\right)^{b(n-k)^\beta}\right) \\
            &\asymp n \cdot n^{\beta - \alpha} \\
            &= n \cdot n^{m_{t^*}-2 - h_{t^*}} \\
            &=  n^{m_{t^*}- 1 - h_{t^*}}.
        \end{align*}

        On the other hand,
        \begin{align*}
            \Delta_{u|t^*} &\geq \frac{1}{n} \cdot \sum_{i,j} \Exp{A_{u|t^*}}_{i,j} \\
            &\geq \frac{1}{n} \cdot \sum_{i,j} \left[ 1-\left(1-\frac{l_{t^*}}{n^{h_{t^*}}}\right)^{N^{t^*}_{i,j}} \right] .
        \end{align*}
        By Assumption \eqref{Assumption:prob_A6},
        \begin{align*}
            N^{t^*}_{i,j} &\geq \xi_{t,i,j} \cdot m_{t^*}! \cdot {{n-2} \choose {m_{t^*}-2}} \\
            &\geq \xi_{t,i,j} \cdot m_{t^*}! \cdot \left(\frac{n-2}{m_{t^*}-2}\right)^{m_{t^*}-2} .
        \end{align*}

        Applying Lemma \ref{cor:binomial} with $k = 2$, $a = l_{t^*}$, $b = \xi_{t,i,j} \cdot m_{t^*}! \cdot (\frac{1}{m_{t^*}-2})^{m_{t^*}-2}$, $\alpha = h_{t^*}$, $\beta = m_{t^*} - 2$, we conclude that 
        \begin{align*}
            \Delta_{u|t^*} &\geq \frac{1}{n} \cdot \sum_{i,j} \left[ 1-\left(1-\frac{l_{t^*}}{n^{h_{t^*}}}\right)^{N^{t^*}_{i,j}} \right] \\
            &\geq \frac{1}{n} \cdot \sum_{i,j} \left[ 1-\left(1-\frac{a}{n^{\alpha}}\right)^{b \cdot (n-2)^\beta} \right] \\
            &\asymp \frac{1}{n} \cdot \sum_{i,j} n^{m_{t^*}-2 - h_{t^*}} \\
            &\asymp n^{m_{t^*}-1 - h_{t^*}}.
        \end{align*}

    \end{enumerate}

    Hence $\Delta_u \ge \Delta_{u|t^*} \in \Theta(n^{\bar \gamma})$ and $\Delta_w \in \Theta(n^{\bar \gamma})$. There exist some scalar $\mu_{T,p}$ independent of $n$, such that
    \begin{equation*} 
        \frac{\Delta_w}{\Delta_u} \leq \mu_{T,p} 
    \end{equation*}
    holds for $n$ sufficiently large and Assumption \eqref{Assumption:P2_A2} holds. Finally, Assumption \eqref{Assumption:P2_A3} is satisfied following Assumption \eqref{Assumption:P1_A1} and Assumption \eqref{Assumption:P2_A2}.
\end{proof}

\subsection{{Proof of Corollary \ref{cor:centrality}}}\label{appendix:centrality}

We start by deriving a convergence result for the normalized adjacency matrix.
\begin{corollary} \label{cor:convergence_SUGM}
    Let $G$ be a random graph of size $n$ generated by the $\text{SUGM}(n,T,p)$ (either weighted or unweighted). Let $A$ be the adjacency matrix of $G$, $\Delta := \norm{\Exp{A}}_{\infty}$ be the maximum expected degree, and $\Bar{A} := A/\Delta$ be the normalized adjacency matrix.  If Assumption \eqref{Assumption:prob_A5} and \eqref{Assumption:prob_A6} are satisfied, then with probability at least $1-\epsilon$,  
    $$\norm{\Bar{A} - \Exp{\Bar{A}}}_2 \leq \rho(n)$$
    where $\rho(n) \rightarrow 0$ as $n \rightarrow \infty$.
\end{corollary}

\begin{proof} 
    By Lemma \ref{lem:prob_frac} Assumption \eqref{Assumption:prob_A5} and \eqref{Assumption:prob_A6} imply that Assumption \eqref{Assumption:P1_A1},  \eqref{Assumption:P2_A2} and \eqref{Assumption:P2_A3} hold. 
    
    For $A_w \sim \text{wSUGM}(n,T,p)$, by Proposition \ref{prop:main_wSUGM_easy}, with probability at least $1-\epsilon$, for $n$ sufficiently large, 
    $$\norm{A_w - \Exp{A_w}}_2 \leq \sqrt{4M^2 \Delta_w\ln{(2n/\epsilon)}}.$$
    Dividing both sides of the inequality by $\Delta_w$, 
    \begin{align*}
        \norm{\Bar{A}_w - \Exp{\Bar{A}}_w}_2 
        &\le \sqrt{\frac{4 M^2 \ln{(2n/\epsilon)}}{\Delta_w}} \rightarrow 0 ~\text{ as } ~n \rightarrow \infty,
    \end{align*}
    where the convergence to zero follows from  Lemma \ref{lem:delta_order}. 

    For $A_u \sim \text{uSUGM}(n,T,p)$, by Proposition \ref{prop:main_uSUGM_ES}, with probability at least $1-\epsilon$, for $n$ sufficiently large, 
    $$\norm{A_u - \Exp{A_u}}_2 \leq 4M^2\sqrt{\mu_{T,p} \Delta_u \ln{(2n/\epsilon)}}.$$
    Dividing both sides of the inequality by $\Delta_u$, 
    \begin{align*}
        \norm{\Bar{A}_u - \Exp{\Bar{A}}_u}_2 &\le  4M^2 \sqrt{\frac{\mu_{T,p} \ln{(2n/\epsilon)}}{\Delta_u}}  \rightarrow 0 ~\text{ as } ~n \rightarrow \infty,
    \end{align*}
    where convergence to zero follows from Assumption \eqref{Assumption:P2_A2} and Lemma \ref{lem:delta_order}.
\end{proof}

\begin{proof} [\textbf{Proof of Corollary \ref{cor:centrality}}]
    We now use Corollary \ref{cor:convergence_SUGM}  to prove  convergence of the centrality measures. All subsequent statements hold with probability $1-\epsilon$.
    
    \begin{enumerate}
        \item For degree centrality: $c^d(\Bar{A}) = \Bar{A} \mathbf{1}$ and $c^d (\Exp{\Bar{A}}) = \Exp{\Bar{A}} \mathbf{1}$. Since $\norm{\mathbf{1}}_2 = \sqrt{n}$, the average difference between the two degree centrality measures 
        \begin{align*}
            &~~~~~ \frac{1}{n} \norm{c^d(\Bar{A}) - c^d(\Exp{\Bar{A}})}_1 
            \\
            &=\frac{1}{n} \norm{\Bar{A} \mathbf{1} -  \Exp{\Bar{A}} \mathbf{1}}_1 = \frac{1}{\sqrt{n}} \frac{\norm{\Bar{A} \mathbf{1} -  \Exp{\Bar{A}} \mathbf{1}}_1}{\sqrt{n}} \\
            &\leq \frac{1}{\sqrt{n}} \norm{\Bar{A} \mathbf{1} -  \Exp{\Bar{A}} \mathbf{1}}_2 \leq \frac{1}{\sqrt{n}} \norm{\Bar{A} -  \Exp{\Bar{A}}}_2 \norm{\mathbf{1}}_2 \\
            &= \norm{\Bar{A} -  \Exp{\Bar{A}}}_2 \le \rho(n)
            \rightarrow 0 \text{ as } n \rightarrow \infty.
        \end{align*}
        
        \item For eigenvector centrality: By Corollary 1 from \cite{yu2015useful} it holds
        \begin{equation*} 
            \norm{v_1(\Bar{A}) - v_1(\Exp{\Bar{A}})}_2 \leq \frac{2^{3/2}\norm{\Bar{A} -  \Exp{\Bar{A}}}_2}{\lambda_1(\Exp{\Bar{A}}) - \lambda_2(\Exp{\Bar{A}})}. 
        \end{equation*} 
        
        Combine it with the definition and assumption \eqref{Assumption:connect}, we get
        \begin{align*}
            &~~~~~ \frac{1}{n} \norm{c^e(\Bar{A}) - c^e (\Exp{\Bar{A}})}_1 \\
            &= \frac{1}{\sqrt{n}} \norm{v_1(\Bar{A}) - v_1(\Exp{\Bar{A}})}_1 \\
            &<\norm{v_1(\Bar{A}) - v_1(\Exp{\Bar{A}})}_2 \leq  \frac{2^{3/2}\norm{\Bar{A} -  \Exp{\Bar{A}}}_2}{\lambda_1(\Exp{\Bar{A}}) - \lambda_2(\Exp{\Bar{A}})} \\
            &=  \frac{2^{3/2}\norm{\Bar{A} -  \Exp{\Bar{A}}}_2}{m} \leq  2^{3/2} \rho(n) \frac{1}{m}\rightarrow 0 \text{ as } n \rightarrow \infty .
        \end{align*}

        \item For Katz centrality: Since $\alpha \in (0,\frac{1}{\lambda^*})$, and $\lambda_1(\Exp{\Bar{A}}) \rightarrow \lambda^*$ by Assumption \eqref{Assumption:connect}, there exist some $N_1 \in \NN$ such that $I - \alpha \Exp{\Bar{A}}$ is invertible for all $n > N_1$ and $c^k_{\alpha}(\Exp{\Bar{A}})) =(I - \alpha \Exp{\Bar{A}})^{-1} \mathbf{1}$ is well defined.

        Under Assumption \eqref{Assumption:prob_A5} and \eqref{Assumption:prob_A6}, $\norm{\Bar{A} -  \Exp{\Bar{A}}}_2 \rightarrow 0$ as $n \rightarrow \infty$ by Corollary \ref{cor:convergence_SUGM}. By Weyl's inequality, $| \lambda_1({\Bar{A}}) - \lambda_1(\Exp{\Bar{A}})| \leq \norm{\Bar{A} -  \Exp{\Bar{A}}}_2$, and since $\alpha < \frac{1}{\lambda^*}$, $\lambda_1({\Bar{A}}) \rightarrow \lambda^*$, there exists some $N_{2} > 0$ such that $\alpha < \frac{1}{\lambda_1({\Bar{A}})}$ for any $n > N_{2}$. This implies that for any $n > N_{2}$, $I - \alpha {\Bar{A}}$ is invertible and $c^k_{\alpha}({\Bar{A}}) = (I - \alpha {\Bar{A}})^{-1} \mathbf{1}$ is well defined. \\

        Note that $\forall \epsilon > 0$, $\exists N_3 \in \NN$ such that $\lambda_1(\Exp{\Bar{A}}) \leq \lambda^* (1+\epsilon)$. In the following, we fix $\epsilon > 0$ such that $\alpha(1+\epsilon) < \frac{1}{\lambda^*}$. Then for $n > \max\{N_1, N_2, N_3\}$,
        \begin{align*}
            \norm{(I - \alpha \Exp{\Bar{A}})^{-1}}_2 &= \frac{1}{1 - \alpha \lambda_1(\Exp{\Bar{A}})} \\
            &\leq \frac{1}{1 - \alpha (1+\epsilon) \lambda^*} =: \frac{1}{\delta}
        \end{align*}
        for $\delta := 1 - \alpha (1+\epsilon) \lambda^* \in (0,1)$ independent of $n$. And 
        \begin{align*}
            &~~~~~ \norm{(I - \alpha \Exp{\Bar{A}}) - (I - \alpha {\Bar{A}})}_2 \\
            &= \alpha \norm{\Bar{A} -  \Exp{\Bar{A}}}_2 \\
            &\leq \alpha \rho(n)\rightarrow 0 \text{ as } n \rightarrow \infty .
        \end{align*}

        It then follows by Theorem 2.3.5 from \cite{han2009theoretical} (with $L := I - \alpha \Exp{\Bar{A}}$ $M: = I - \alpha {\Bar{A}}$), that 
        \begin{align*}
            &~~~~~ \norm{(I \!-\! \alpha \Exp{\Bar{A}})^{-1} \!-\! (I \!-\! \alpha {\Bar{A}})^{-1}}_2 \\
            &\leq \frac{\norm{(I \!-\! \alpha \Exp{\Bar{A}})^{-1}}_2^2 \norm{(I \!-\! \alpha \Exp{\Bar{A}}) \!-\! (I \!-\! \alpha {\Bar{A}})}_2}{1\!-\! \norm{(I \!-\! \alpha \Exp{\Bar{A}})^{-1}}_2 \norm{(I \!-\! \alpha \Exp{\Bar{A}}) \!-\! (I \!-\! \alpha {\Bar{A}})}_2} \\
            &\le \frac{ \alpha\rho(n) / \delta^2}{1 \!-\! \alpha\rho(n) / \delta} \\
            &\leq \frac{2\alpha}{\delta^2} \rho(n)
        \end{align*}

        where we used $1 - \alpha\rho(n) / \delta \geq 1/2$ for large $n$.
        And finally, since $\norm{\mathbf{1}}_2 = \sqrt{n}$, we get the average difference between the two Katz centrality measures is
        \begin{align*}
            &~~~~~\frac{1}{{n}} \norm{c^k_{\alpha} (\Bar{A}) - c^k_{\alpha} (\Exp{\Bar{A}})}_1 \\
            &=
            \frac{1}{\sqrt{n}} \frac{\norm{(I - \alpha \Exp{\Bar{A}})^{-1} \mathbf{1} - (I - \alpha {\Bar{A}})^{-1} \mathbf{1}}_1}{\sqrt{n}} \\
            &\leq \frac{1}{\sqrt{n}} \norm{(I - \alpha \Exp{\Bar{A}})^{-1} \mathbf{1} - (I - \alpha {\Bar{A}})^{-1} \mathbf{1}}_2 \\
            &\leq \frac{1}{\sqrt{n}} \norm{(I - \alpha \Exp{\Bar{A}})^{-1} - (I - \alpha {\Bar{A}})^{-1}}_2 \sqrt{n} \\
            &\leq \frac{2\alpha}{\delta^2} \rho(n) \rightarrow 0 \text{ as } n \rightarrow \infty.
        \end{align*}
    \end{enumerate}
\end{proof}

\bibliographystyle{IEEEtran}
\bibliography{reference.bib}

\end{document}